\documentclass[11pt,english]{article}
\usepackage[T1]{fontenc}
\usepackage[latin9]{inputenc}
\usepackage{geometry}
\usepackage{float}
\usepackage{units}
\usepackage{amsmath}
\usepackage{amsthm}
\usepackage{graphicx}
\usepackage[authoryear]{natbib}
\usepackage{pdfsync}
\usepackage{color}
\usepackage{mdframed}

\makeatletter

\providecommand{\tabularnewline}{\\}
\floatstyle{ruled}
\newfloat{algorithm}{tbp}{loa}
\providecommand{\algorithmname}{Algorithm}
\floatname{algorithm}{\protect\algorithmname}

\theoremstyle{plain}
\newtheorem{thm}{\protect\theoremname}
  \theoremstyle{plain}
  \newtheorem{prop}[thm]{\protect\propositionname}

\date{}
\usepackage{placeins}
\usepackage{pdflscape}
\usepackage[T1]{fontenc}

\@ifundefined{showcaptionsetup}{}{%
 \PassOptionsToPackage{caption=false}{subfig}}
\usepackage{subfig}
\makeatother

\usepackage{babel}
  \providecommand{\propositionname}{Proposition}
\providecommand{\theoremname}{Theorem}

\begin{document}

\title{A Decomposition Algorithm to Solve the Multi-Hop Peer-to-Peer Ride-Matching
Problem}

\maketitle
\begin{center}
Transportation Research Part B: Methodological\\
May 2017\\
\end{center}

\vspace{2pc}
\textbf{Neda Masoud }

Assistant Professor

Department of Civil and Environmental Engineering

University of Michigan Ann Arbor, Ann Arbor, MI, USA, 48109 

nmasoud@umich.edu 

Corresponding author \\

\textbf{R. Jayakrishnan }

Professor 

Department of Civil and Environmental Engineering

University of California Irvine, Irvine, CA, USA, 92697

rjayakri@uci.edu 

\newpage{}
\begin{abstract}
\noindent In this paper, we mathematically model the multi-hop Peer-to-Peer
(P2P) ride-matching problem as a binary program. We formulate this
problem as a many-to-many problem in which a rider can travel by transferring between multiple drivers, and a driver can carry multiple riders. We propose
a pre-processing procedure to reduce the size of the problem, and devise a decomposition algorithm to solve the original ride-matching
problem to optimality by means of solving multiple smaller problems.
We conduct extensive numerical experiments to demonstrate the computational
efficiency of the proposed algorithm and show its practical applicability
to reasonably-sized dynamic ride-matching contexts. Finally, in the
interest of even lower solution times, we propose heuristic solution
methods, and investigate the trade-offs between solution time and
accuracy.
\end{abstract}

\section{Introduction}

Recent advances in communication technology coupled with increasing
environmental concerns, road congestion, and the high cost of vehicle
ownership have directed more attention to the opportunity cost of
empty seats traveling throughout the transportation networks every
day. Peer-to-peer (P2P) ridesharing is a good way of using the existing
passenger-movement capacity on the vehicles, thereby addressing the
concerns about the increasing demand for transportation that are too
costly to address via infrastructural expansion.

Although limited versions of P2P ridesharing systems initially emerged
in the US in the 1990s, they did not receive enough support from the
targeted population to continue operating. Inadequate and non-targeted
marketing, insufficient flexibility and convenience, and absence of
appropriate technology were some of the factors that contributed to
the lack of success in implementation of the first generation
of ridesharing systems. 

A second generation of ridesharing systems emerged after considerable
improvements in communications technology in the past few years. Making
use of GPS-enabled cell phones in more recent ridesharing systems
allows for accessing online information on the location of participants,
making ridesharing more accessible and providing people with a sense
of security. These factors combined with the high cost of travel (both
financial and environmental) have played an important role in the
higher levels of interest in ridesharing systems in the recent years. The high
growth rate of Transportation Network Companies (TNC) such as Uber
and Lyft in the USA and elsewhere in the world is an indicator of
increasing levels of acceptance in the concept of outsourcing rides.

In line with the demand side, the supply side of P2P ridesharing has
experienced growth as well. According to the 2009 national household
travel survey (NHTS), out of an average of 4 seats available in a
vehicle, only 1.7 is being actually used (Figure \ref{Fig:Occupancy}).
This number is as low as 1.2 for work-based trips. In addition, number of
trips per household in the US has been experiencing a decreasing trend since
1995. On the contrary, the general trend in the number of vehicles
owned by households has been increasing. These statistics suggest
that carpooling in households is declining, and that the number of
empty seats available on traveling vehicles is increasing. 

This increase on the supply side of ridesharing, coupled with the
rise on the demand side, imply an optimistic future for P2P ridesharing
services. This potential was recognized by the US congress in June
2012. Section 1501 of the Moving Ahead for Progress in the 21st Century
(MAP-21) transportation act expanded the definition of ``carpooling''
to include ``real-time ridesharing'' as well, making ridesharing eligible for  the federal funds that were previously available only for carpooling projects.

In this paper, we define P2P dynamic ridesharing to include all one-time
rideshares with any type of arrangement, whether it is on-the-fly
or pre-arranged, between peer drivers and riders. Dynamic ridesharing
differs from more traditional carpooling services in that in carpooling
shared trips are scheduled for an extended period of time, and are
not one time occurrences. Furthermore, the nature of carpooling programs
does not ask for real-time ride-matching. 

Drivers in a ridesharing system drive to perform activities of their
own, and not for the mere purpose of transporting riders. Each driver
can have multiple riders on board at any point in time. In addition,
to increase the number of served rider requests, the system provides
multi-hop itineraries for riders, where riders may transfer between
vehicles. 

We define a set of stations in the network where riders and drivers
can start and end their trips, and riders can transfer between vehicles.
The system finds matches for riders by optimally routing drivers in
the network. In order to guarantee a high quality of service, both
riders and drivers provide a time window to specify the start and
end of their trips, and a maximum ride time. Furthermore, riders can
specify the maximum number of transfers they are willing to make,
and drivers can put a limit on the number of riders they want to have
on board at each moment in time. The term ``real-time''
emphasizes the capability of the system to make ride-matches in a
short period of time, for implementation with frequent re-optimizations
using newer data over time. 

A P2P ride-matching algorithm is central to successful implementation
of a ridesharing system. Ride-matching refers to the problem of matching
riders (passengers) and drivers in a ridesharing system. A successfully
matched rider receives from the system an itinerary of his/her trip
that includes information on the scheduled route, and the drivers
with whom the travel is planned. Drivers receive itineraries that
include the schedules to pick up and drop off riders. 

\begin{center}
\begin{figure}
\begin{centering}
\includegraphics[width=5in]{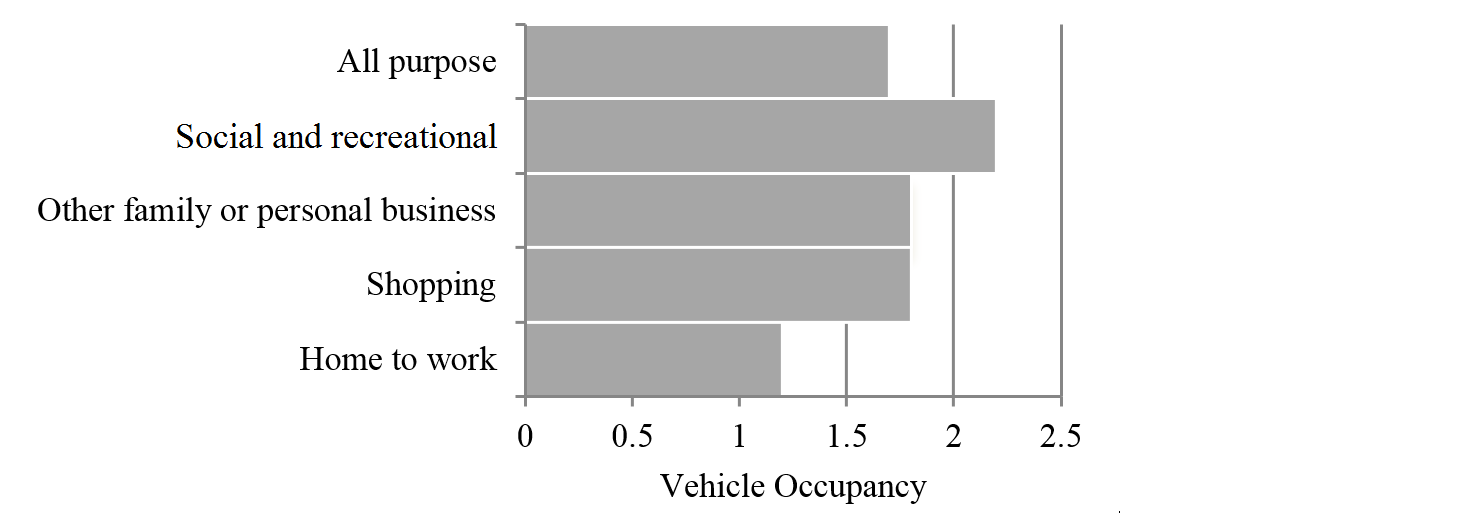}
\par\end{centering}
\caption{Average vehicle occupancy in the US in 2009 for different trip purposes
(\citealp{NHTS})}
\label{Fig:Occupancy}
\end{figure}
\par\end{center}

\section{Literature Review\label{sec:Literature-Review}}

The Multi-hop P2P ride-matching problem can be formulated as a special
case of the general pick-up and delivery problem (GPDP). GPDP consists
of devising a set of routes to satisfy transportation requests with
given loads, and origin/destination locations. Vehicles that operate
these routes each have a certain origin, destination, and capacity
(\citet{1995_GPDP}). 

The dial-a-ride problem (DARP) is a special case of the GPDP, where
all vehicles share the same origin and destination depot, and the
loads to be transported are people. Although DARP is usually used
in systems that aim at transporting elderly or handicapped people,
this problem is very close to the ride-matching problem in ridesharing
systems. 

In its basic form, DARP considers a depot where a fleet of homogeneous
vehicles start their trips in the morning, and to which they return
at the end of their shifts. Each passenger is assumed to make the
entire trip in the same vehicle, i.e., the possibility of transfers
between vehicles is not considered. Variants of DARP that are more
application-friendly consider time windows for the pick-up and delivery
of passengers (\citealp{DARPTW1,DARPTW2}). \citet{2007_Multi_Depot1,Overview_DARP_2007} provide an overview of the literature on DARP. 

In reality, the problem of transporting passengers is often more complex
than the basic form of DARP. Some agencies have their fleet located
at stations throughout their operating area. This has motivated the
development of the multi-depot formulation for DARP (MD-DARP) (\citealt{2007_Multi_Depot1}).
Recently, \citet{MD_H_DARP1} and \citet{MD_H_DARP2} have added heterogeneity
into the mix, and studied the multi-depot heterogeneous DARP (MD-H-DARP).
These studies include heterogeneity among vehicles (multiple depots,
capacity, level of service, and operating costs) as well as passengers
(level of care, accompanying individuals, required resources, and
total number of passengers). 

An additional degree of flexibility that has recently been added to
the original DARP is the possibility for passengers to transfer between
multiple vehicles/modes of transportation, leading to the emergence
of the DARP with transfers (DARPT). \citeauthor{DARP_Transfer1}~(2014a),
\citet{DARP_Transfer2}, and \citet{DARP_transfer_bimodal} are the
only papers that have studied this variant of the problem, to the
best of our knowledge. \citeauthor{DARP_Transfer1}~(2014a) limit
the number of potential transfers to one. \citet{DARP_Transfer2}
does not put a constraint on the capacity of vehicles, and works with
demand at an aggregate level, rather than the individual passengers'
travel desires. \citet{DARP_transfer_bimodal} use heuristic algorithms
to propose multi-modal routes to para-transit users. In their study,
they try to route para-transit vehicles to carry passengers from their
homes to bus stops, and from bus stops to their destinations. 

The P2P ride-matching problem has attracted attention in academia
only in the very recent years. Ride-matching problems share some of
the characteristics of the more advanced DARPs, such as multiple depots
and heterogeneous vehicles and passengers. Drivers in ridesharing
systems are traveling to perform activities, and have distinct origin
and destination locations (multi-depot), different vehicle capacities
(heterogeneity), and rather narrow travel time windows. These factors
can lead to the matching problems in ridesharing systems being spatiotemporally
sparse, in general. One characteristic that differentiates the ride-matching
problem from DARP is the fact that the set of vehicles in a ridesharing
system is neither fixed (i.e., not a certain fleet size is available
on a regular basis) nor deterministic (i.e., the system does not know
in advance the time windows and origins and destinations of drivers'
trips). In addition, drivers who make their vehicles available in
ridesharing systems are peers to the passengers who are looking for
rides, and therefore measures of quality of service that are reserved
only for passengers in DARPs should be extended to drivers as well
in ridesharing systems.

\citet{agatz2012optimization} and \citet{furuhata2013ridesharing}
classify ride-sharing systems based on different criteria, and discuss
the challenges ridesharing systems face. In its simplest form, the
ride-matching problem matches each driver with a single rider. This
can be modeled as a maximum-weight bipartite matching problem that
minimizes the total rideshare cost (\citealp{Agatz_Rider_or_Driver}). 

There are also ride-matching problems that are more complex and try
to take advantage of the full unused capacity of vehicles by allowing
multiple riders in each vehicle. This form of ridesharing is similar
to the carpooling problem where a large employer encourages its employees
to share rides to and from work (\citealp{carpool_exact} and \citealp{2004_carpool_hueristic}).
The taxi-sharing problem, as formulated by \citet{hosni2014shared},
also tries to reduce the cost of taxi services by having people share
their rides. \citet{Herbawi1} have studied the problem of matching
one driver with multiple riders in the context of ridesharing, and
have proposed non-exact evolutionary multi-objective algorithms. \citet{Febbraro}
formulate an optimization problem to model the many-to-one ridesharing
systems (in which each rider is paired with only one driver, though each driver can carry multiple riders), and use optimization engines to solve it. \citet{stiglic2015benefits} manage to increase the number of served riders by having riders walk to meeting points, where multiple riders can be picked up by a driver. The number of stops for each driver, however, is limited to a maximum of two.

\citet{Herbawi_1RMD_1,Herbawi_1RMD_2} model another variant of the ride-matching problem in which a single rider can travel by transferring between multiple drivers. They propose a genetic algorithm to solve this one-to-many matching problem. \citeauthor{masson2014optimization}~(2014b)
study a similar problem in a multi-modal environment, where goods are carried using a combination of excess bus capacities and city freighters. They propose an adaptive large neighborhood heuristic algorithm to solve the problem. \citet{coltin2014ridesharing} show the fuel efficiency that including transfers in the riders' itineraries can offer, using three heuristic algorithms.

Many-to-many matching problems allow drivers to have multiple passengers on board at each point in time, and riders to transfer between drivers \citep{Agatz}. \citet{cortes2010pickup} were the first to formally
formulate a many-to-many pick-up and delivery problem. They introduced an exact branch-and-cut solution method. The largest example they solved, however, consisted of 6 requests, two vehicles, and one transfer point. To the best of our knowledge, there are only two studies that model many-to-many ridesharing systems. \citet{Agatz} is one of the
first to take an optimization approach toward modeling many-to-many ridesharing systems. In their study, the authors discuss modeling multi-modal ridesharing systems that allow for transfers between different modes of transport. However, they do not discuss a solution methodology.
\citet{Ghoseiri} formulates a mixed integer problem (MIP) to model the many-to-many ridesharing system. Their proposed solution heuristic limits the number of transfers to a maximum of two. 

In addition to multi-hop ridesharing, a solution methodology that
can handle unlimited number of transfers can be used to optimize multi-modal
transportation networks, where, for example, ridesharing is combined
with public transportation. In such a scenario, each public transport
line acts as a driver with a fixed route in the ridesharing system.
In case of a multi-modal network, allowing for higher numbers of transfers and devising algorithms that can accomplish routing and scheduling of passengers in real-time become an integral part of the system.

In this paper, we model a multi-hop P2P ridesharing system as a binary
optimization problem, and propose an efficient algorithm to solve
the corresponding ride-matching problem to optimality. Our model allows
drivers to carry multiple riders at the same time, and riders to transfer
between vehicles, and/or different modes of transportation. This leads
to higher system performance, which can, to some degree, compensate
for the inherent spatiotemporal sparsity in ridesharing systems. In
contrary to some ridesharing systems who either assume a given route
for drivers, or allow for only limited detours from a given route,
we leave the routing of the drivers to the system as a default, unless
drivers specifically ask to take certain routes. In order to ensure
that the higher performance of the system does not come at the cost
of poor quality of service, we require both drivers and riders to
have specified a maximum ride time, and for riders to have set a limit
on the number of transfers. To ensure that rides can be successfully
accomplished within the requested time windows specified by participants,
the formulation is devised to use time-dependent travel time matrices.

The contributions of this paper to the literature are three-fold;
we formulate the many-to-many ride-matching problem as a binary program
in a time-expanded network, introduce a pre-processing procedure to
reduce the size of the proposed optimization problem, and devise a
decomposition algorithm to solve the large-scale ride-matching problem
to optimality by means of iteratively solving smaller problems. The combination
of the three aforementioned methodological contributions allows for
the ride-matching problem to be solved in a very short period of time,
enabling dynamic implementation of multi-hop ridesharing. 

In the rest of the paper, we first introduce the ridesharing system
to provide the context in which we need to solve the ride-matching
problem. The rest of the paper focuses on formulating and solving
the ride-matching problem. We start by formulating the multi-hop P2P
ride-matching problem as an optimization problem. Since solving this problem directly using optimization engines is computationally prohibitive, we introduce a pre-processing procedure to limit the
size of the input sets to the optimization problem, and a decomposition
algorithm to solve the problem more efficiently in terms of computing
time. Finally we address the scalability of the problem and real-time implementation strategies for the underlying system, and introduce heuristics to speed up reaching a high quality solution.

\section{Ridesharing System\label{sec:Ridesharing-System}}

The ridesharing system defined in this paper contains a set of participants
$P$. These participants are divided into a set of riders, $R$, who
are looking for rides, and a set of drivers, $D$, who are willing
to provide rides ($P=R\cup D$). Drivers may have different incentives
to participate in the ridesharing system, including monetary compensation,
using high occupancy vehicle (HOV) lanes, or reduced-cost parking,
among others. Different driver incentives for ridesharing can translate
into different objective functions for the corresponding ride-matching
problem. 

To facilitate pick-ups and drop-offs, a set of stations, $S$, are
identified in the network. Stations are pre-specified locations where
participants can start and end their trips, and riders can switch
between drivers and/or to other modes of transport, such as transit.
Strategic identification of stations is central to the performance
of the system. Lessons learned from the previous P2P ridesharing systems
suggest that it is better for riders to be picked up/dropped off at
pre-specified locations, rather than their homes (or the exact location
where their trips start/end) for two reasons (\citealp{heinrich2010implementing}).
First, these locations could be hard to find for drivers, and therefore
riders could miss their scheduled rides. In addition, drivers could
have a hard time finding an appropriate location to park their vehicles.
Second, some drivers and riders would understandably be reluctant
to reveal their home address to others. In addition, as shown by \citet{stiglic2015benefits},
introducing stations can increase the number of successful matches, as a result of increasing the spatial proximity between trips.

Every active participant $p$ of the rideshare system provides to
the system their origin and destination stations ($OS_{p}$ and $DS_{p}$, respectively), the earliest acceptable time to depart from the origin
station, $T_{p}^{ED}$, the latest acceptable time to arrive at the
destination station, $T_{p}^{LA}$, and their maximum ride time, $T_{p}^{TB}$.
Subsequently, the travel time window for participant $p$ can be defined as $TW_{p}=[T_{p}^{ED},T_{p}^{LA}]$.
In addition, each participant is asked to provide a notification deadline
by which they need to be informed of any matches made for them. Figure \ref{fig:story} demonstrates some of the parameters of the system.

\begin{figure}[H]
\begin{centering}
\includegraphics[width=4in]{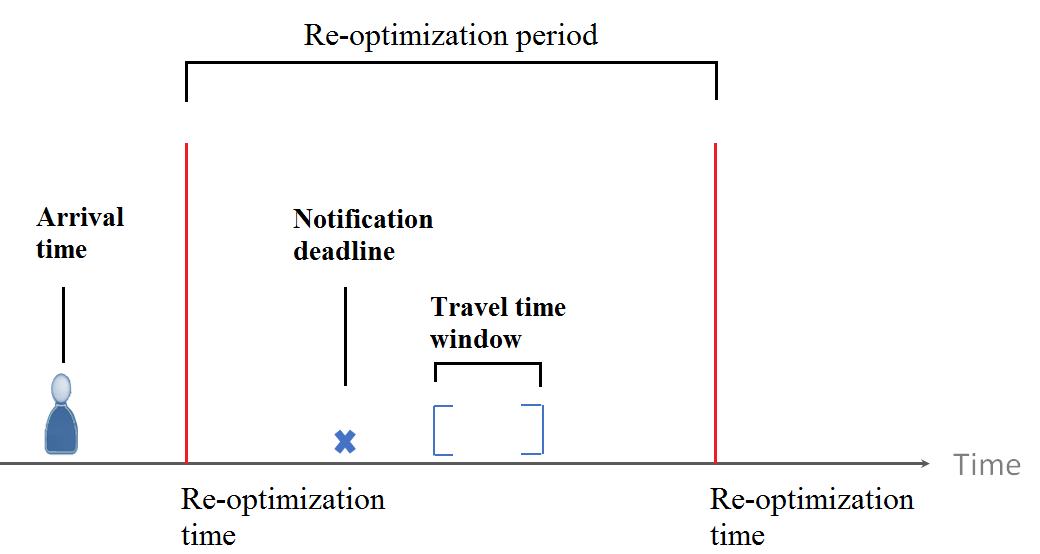}
\par\end{centering}
\caption{System parameters\label{fig:story}}
\end{figure}

Drivers should announce the capacity of their vehicles, $C_{d}$.
This could simply be the physical vehicle capacity, or the
maximum number of riders a driver is willing to carry in their vehicle
at any point in time. Riders can specify the maximum number of transfers
(change of vehicles), $V_{r}$, they are willing to make to get to
their destinations. 

We discretize the study time horizon into an ordered set of indexed
time intervals of small duration, $\Delta t$, to allow for using
time-dependent travel-time matrices. We define set $T_{p}$ to contain
the indices for all time intervals within the range $TW_{p}=[T_{p}^{ED},T_{p}^{LA}]$
for participant $p$.

In a system discretized in both time and space, we define a node $i$,
$n_{i}$, as a tuple $(t_{i},s_{i})$, where $t_{i}$ is the time
interval one may be located at station $s_{i}$. Subsequently, a link
$\ell$ is defined as $(n_{i},n_{j})=(t_{i},s_{i},t_{j},s_{j})$,
where $t_{i}$ is the time interval one has to leave station $s_{i}$,
in order to arrive at station $s_{j}$ during time interval $t_{j}$.
We generate links between neighboring stations only, i.e., a link exists
between stations $s_{i}$ and $s_{j}$ if traveling from $s_{i}$
to $s_{j}$ does not require passing through another station. Naturally,
the links are determined by the travel time matrix, which could be
time-dependent, and travel times used must ensure that a traveler
leaving at the very end of the time interval of $t_{i}$ can arrive
at least by the very end of the time interval $t_{j}$. We denote
the set of links by $L$. 

Figure \ref{fig:Network_exp} demonstrates an example of a link in
a time-expanded network. If a participant leaves station 1 at time
interval $1$, they arrive at station $2$ at time interval $2$.
The starting and ending nodes of this trip are $n_{s}=(t_{s},s_{s})=(1,1)$, and $n_{e}=(t_{e},s_{e})=(2,2)$, respectively. The resulting link
is $\ell=(t_{s},s_{s},t_{e},s_{e})=(1,1,2,2)$. In this figure we
can also see link $(2,2,3,4)$. These two links together form
a path that leaves station $1$ at time interval $1$, and arrives
at station $4$ at time interval $3$.

\begin{figure}
\begin{centering}
\includegraphics[width=3in]{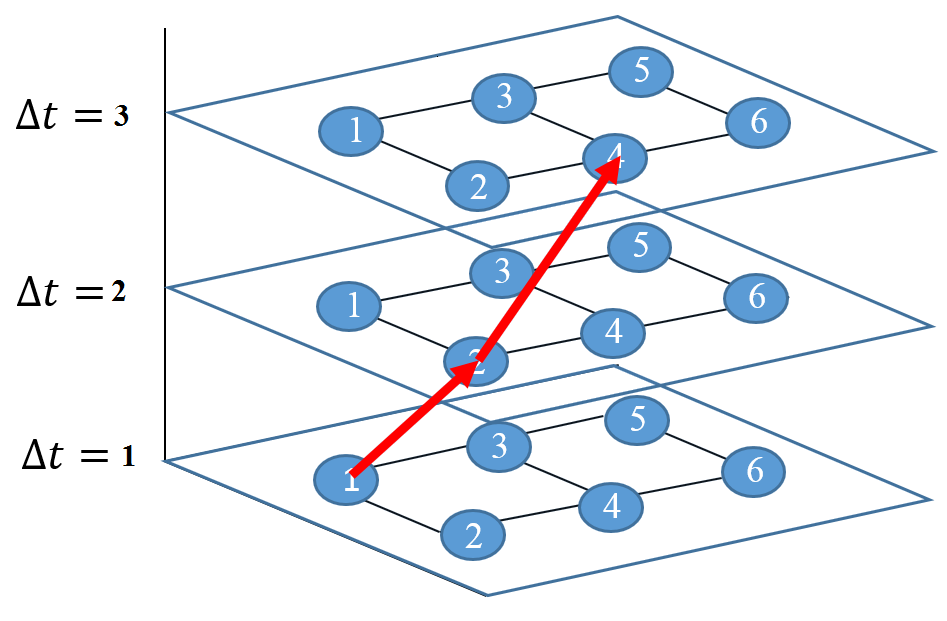}
\par\end{centering}
\caption{Example of links in a time-expanded network\label{fig:Network_exp}}
\end{figure}

The goal of the ridesharing system in this paper is to maximize the
matching rate (although the proposed methodology can be used for a variety of objective functions.) For now, let us assume that drivers leave the route
choice to the system. This does not preclude the case of drivers who
want to follow their own fixed routes, as those routes can be specified
as successive nodes and entered into the formulation as fixed parameters.

We include a dummy driver, $d^{\prime}$, in the set of drivers, and
form the set $D^{\prime}=D \cup \{d^{\prime}\}$. The dummy driver
has neither a real origin or destination, nor a travel time window.
The motivation behind introducing the dummy driver will be explained
in the following section.

The implementation strategy of a ridesharing system should be determined based on the nature of customer arrivals. Here, we adapt a general \textit{rolling time horizon} framework in which the system is re-optimized periodically at pre-specified points in time called ``re-optimization times''. We call the time window between two consecutive re-optimization times the ``re-optimization period'' (Figure \ref{fig:story}). Re-optimizing the system allows for enhancing the reliability of the system and improving its performance by incorporating the latest information on link travel times and customer arrivals, respectively. In a rolling time horizon framework, at each re-optimization time the system solves a matching problem that includes all registered participants whose notification deadlines are after the re-optimization time. In addition, all previously matched drivers who have not finished their trips at the re-optimization time can be included in the matching problem as well, albeit if they have empty seats left, and considering their prior commitments. After solving the problem at each re-optimization time, the itineraries of matched drivers and riders whose notification deadlines lie within the first re-optimization period will be fixed and announced to them. 

As previously mentioned, the rolling time horizon framework is a general framework whose parameters (i.e., the re-optimization period and re-optimization times) can be calibrated based on the distribution of customer arrivals. In an extreme case where all system participants register their trips before the onset of the planning horizon (e.g., before the start of a morning peak hour), the general rolling time horizon framework turns into a \textit{static} framework. On the other side of the spectrum, if all participants arrive in real-time (i.e., their arrival times and notification deadlines are less than a few minutes apart), the rolling time horizon framework can be used to serve such a \textit{dynamic} system by selecting very small re-optimization periods (e.g., 1-min periods).

\section{The Multi-hop P2P Ride-Matching Problem}

The role of the ride-matching problem is to devise itineraries that
can take riders to their destinations by optimally routing drivers.
Itineraries have to comply with the riders' specified
maximum number of transfers, the capacity of drivers'
vehicles, and the participants' travel time windows and maximum
ride times. The ride-matching problem will devise itineraries for
the matched riders in the system, and all drivers, matched or not. 

To mathematically model the ride-matching problem, we use four sets
of decision variables, as defined in (\ref{eq:var_Xd})-(\ref{eq:Var_U}). 

\begin{align}
x_{\ell}^{d} & =\begin{cases}
1 & \mbox{Driver \ensuremath{d} travels on link \ensuremath{\ell}}\\
0 & \mbox{Otherwise}
\end{cases}\label{eq:var_Xd}\\
y_{\ell}^{rd} & =\begin{cases}
1 & \mbox{Rider \ensuremath{r} travels on link \ensuremath{\ell} with driver \ensuremath{d}}\\
0 & \mbox{Otherwise}
\end{cases}\label{eq:var_Xrd}\\
z_{r} & =\begin{cases}
1 & \mbox{Rider \ensuremath{r} is matched}\\
0 & \mbox{\mbox{Otherwise}}
\end{cases}\label{eq:Var_Y}\\
u_{r}^{d} & =\begin{cases}
1 & \mbox{Driver \ensuremath{d} contributes to the itinerary for rider \ensuremath{r}}\\
0 & \mbox{\mbox{Otherwise}}
\end{cases}\label{eq:Var_U}
\end{align}

Equation (\ref{eq:Obj}) presents the objective function of the problem.
A ride-matching problem can have various objectives, ranging from
maximizing profits to minimizing the total miles/hours traveled in
the network. This objective can vary depending on the nature of the
agency that is managing the system (public or private), the level
of acceptance of the system in the target community, and the ridesharing
incentives. For a ridesharing system in its infancy, it is logical
to maximize the matching rate. We use this objective for the ridesharing
system in this paper. 

The first term in (\ref{eq:Obj}) maximizes the total number of served
riders, while the second term minimizes the total number of transfers
in the system, and is added only for the purposes described in Proposition
1 (in the Appendix). The weight $W_{r}$ should be set to a proper
value (any value smaller than $\nicefrac{1}{V_{r}}$, where $V_{r}$
is the maximum number of transfers rider $r$ is willing to make)
to ensure that serving the maximum number of riders remains the primary
objective of the system. 

The sets of constraints that define the ridesharing system are presented
in (\ref{eq:d_orig})-(\ref{eq:transfer}). Constraint sets (\ref{eq:d_orig})-(\ref{eq:d_balance})
route drivers in the network. Constraint set (\ref{eq:d_orig}) directs
drivers in set $D$ out of their origin stations, and (\ref{eq:d_dest})
ensures that they end their trips at their destination stations. Note
that we do not use a separate set of constraints to enforce the travel
time windows for drivers. Such constraints can be satisfied automatically
by limiting the set of links in constraint sets (\ref{eq:d_orig})
and (\ref{eq:d_dest}) to the ones whose time intervals are within
the driver's travel time window (in set $T_{d})$. 

Constraint set (\ref{eq:d_balance}) is for flow conservation, enforcing
that a driver entering a station in a time interval, exits the station
in the same time interval. Notice that participants might not physically
leave a station. Members of the link set $L$ in the form $(t,s,t+1,s)$
represent the case where a participant is physically remaining at
station $s$ for one time interval, but technically leaving node $(t,s)$
for node $(t+1,s)$. In addition, constraint set (\ref{eq:d_balance})
enforces the multi-hop property of the ridesharing system. Riders
can enter a node with a driver, and exit it with a different driver,
suggesting a transfer between the two drivers (or modes of transportation).
Constraint set (\ref{eq:d_travel_time}) limits the total travel time
by drivers based on their maximum ride times.

\begin{subequations}\label{equ:determinstic}

\begin{align}
\mbox{Max}\quad & \;\;\sum_{\medspace\medspace\medspace r\in R}z_{r}-\sum_{r\in R}W_{r}\sum_{d\in D}u_{r}^{d}\label{eq:Obj}\\
\mbox{Subject to:}\quad & \sum_{^{_{s_{i}=OS_{d}(t_{i},t_{j})\in T{}_{d}}^{\quad\:\;\;\quad\ell\in L:}}}\!\!\!\!\!\!\!\!\!\!x_{\ell}^{d}\:\:\:\:-\quad\!\!\!\!\!\!\!\!\!\sum_{^{_{s_{j}=OS_{d};t_{i},t_{j}\in T{}_{d}}^{\quad\:\;\;\quad\ell\in L:}}}\!\!\!\!\!\!\!\!\!\!x_{\ell}^{d}=1 &  & \forall d\in D\label{eq:d_orig}\\
 & \sum_{^{_{s_{j}=DS_{d};t_{i},t_{j}\in T{}_{d}}^{\quad\:\;\;\quad\ell\in L:}}}\!\!\!\!\!\!\!\!\!\!x_{\ell}^{d}\:\:\:\:-\quad\!\!\!\!\!\!\!\!\!\sum_{^{_{s_{i}=DS_{d};t_{i},t_{j}\in T{}_{d}}^{\quad\:\;\;\quad\ell\in L:}}}\!\!\!\!\!\!\!\!\!\!x_{\ell}^{d}=1 &  & \forall d\in D\label{eq:d_dest}\\
 & \;\;\sum_{_{\ell=(t_{i},s_{i},t,s)\in L}^{\qquad t_{i},s_{i}}}\!\!\!\!\!x_{\ell}^{d}\:\:\:=\!\!\!\!\!\!\sum_{_{\ell=(t,s,t_{j},s_{j})\in L}^{\qquad t_{j},s_{j}}}\!\!\!\!\!x_{\ell}^{d} &  & \!\!\!\begin{array}{l}
\forall d\in D\\
\forall t\in T_{d}\\
\forall s\in S\backslash\{OS_{d}\cup DS_{d}\}
\end{array}\label{eq:d_balance}\\
 & \;\;\sum_{^{\ell\in L}}(t_{j}-t_{i})x_{\ell}^{d}\leq\frac{T_{d}^{TB}}{\Delta t} &  & \forall d\in D\label{eq:d_travel_time}\\
 & \:\:\sum_{^{d\in D^{\prime}}}\!\!\!\sum_{^{_{s_{i}=OS_{r};t_{i},t_{j}\in T{}_{r}}^{\quad\:\;\;\quad\ell\in L:}}}\!\!\!\!\!\!y_{\ell}^{rd}-\sum_{^{\quad d\in D^{\prime}}}\!\!\!\sum_{^{_{s_{j}=OS_{r};t_{i},t_{j}\in T{}_{r}}^{\quad\:\;\;\quad\ell\in L:}}}\!\!\!\!\!\!y_{\ell}^{rd}=z_{r} &  & \forall r\in R\label{eq:r_orig}\\
 & \:\:\sum_{^{d\in D^{\prime}}}\!\!\!\sum_{^{_{s_{j}=DS_{r};t_{i},t_{j}\in T{}_{r}}^{\quad\:\;\;\quad\ell\in L:}}}\!\!\!\!\!\!y_{\ell}^{rd}-\sum_{^{d\in D^{\prime}}}\!\!\!\sum_{^{_{s_{i}=DS_{r};t_{i},t_{j}\in T{}_{r}}^{\quad\:\;\;\quad\ell\in L:}}}\!\!\!\!\!\!y_{\ell}^{rd}=z_{r} &  & \forall r\in R\label{eq:r_dest}\\
 & \:\:\sum_{^{d\in D^{\prime}}}\!\!\!\sum_{^{_{\ell=(t_{i},s_{i},t,s)\in L}^{\qquad t_{i},s_{i}:}}}\!\!\!\!\!y_{\ell}^{rd}=\sum_{^{\quad d\in D^{\prime}}}\!\!\!\sum_{^{_{\ell=(t,s,t_{j},s_{j})\in L}^{\qquad t_{j},s_{j}:}}}\!\!\!\!\!y_{\ell}^{rd} &  & \!\!\!\begin{array}{l}
\forall r\in R\\
\forall t\in T_{r}\\
\forall s\in S\backslash\{OS_{r}\cup DS_{r}\}
\end{array}\label{eq:r_balance}\\
 & \:\:\sum_{d\in D^{\prime}}\:\:\:\:\sum_{^{\ell\in L}}(t_{j}-t_{i})y_{\ell}^{rd}\leq\frac{T_{r}^{TB}}{\Delta t} &  & \forall r\in R\label{eq:r_travel_time}\\
 & \:\:\sum_{r\in R}y_{\ell}^{rd}\leq C_{d}x_{\ell}^{d} &  & \!\!\!\begin{array}{l}
\forall d\in D\\
\forall\ell\in L
\end{array}\label{eq:capacity}\\
 & \;\;u_{r}^{d}\geq y_{\ell}^{rd} &  & \!\!\!\begin{array}{l}
\forall r\in R\\
\forall d\in D\\
\forall\ell\in L
\end{array}\label{eq:ident_driver_1}\\
 & \;\;u_{r}^{d}\leq\sum_{\ell\in L}y_{\ell}^{rd} &  & \!\!\!\begin{array}{l}
\forall r\in R\\
\forall d\in D
\end{array}\label{eq:ident_driver_2}\\
 & \:\:\sum_{d\in D}u_{r}^{d}-1\leq V_{r} &  & \forall r\in R\label{eq:transfer}
\end{align}

\end{subequations}

Rider $r$'s itinerary is determined by variable $y_{\ell}^{rd}$.
A value of $1$ for this variable indicates that the rider is traveling
on link $\ell$ in driver $d$'s vehicle. By definition,
this variable implies that a rider should always be accompanied by
a driver. However, in reality, a rider does not need to be accompanied
when he/she is traveling on a link in the form $(t,s,t+1,s)$, i.e.,
staying at station $s$, waiting to make a transfer. To incorporate
this element into the formulation, we introduce the dummy driver $d^{\prime}$.
As mentioned before, the dummy driver does not have a real origin
or destination in the network. The set of links used by the dummy
driver is also different from the members of the link set $L$. We
define the set of links for the dummy driver as $L^{\prime}=\{(t,s,t+1,s),\forall(t,s)\in T\times S\}$.
This set includes all the links that represent staying at a station
for one time interval. 

Constraint sets (\ref{eq:r_orig})-(\ref{eq:r_balance}) route riders
in the network, and are analogous to (\ref{eq:d_orig})-(\ref{eq:d_balance}),
except for a small variation. While the optimization problem generates itineraries for all drivers, matched or not, this is not the case for riders. Only riders who are successfully matched will receive itineraries. This
difference is reflected in the formulation by replacing $1$ on the
ride hand side of constraint sets (\ref{eq:d_orig})-(\ref{eq:d_dest})
by $z_{r}$ in constraint sets (\ref{eq:r_orig})-(\ref{eq:r_dest}).
Constraint set (\ref{eq:r_travel_time}) sets a limit on the riders'
maximum ride times. 

Constraint set (\ref{eq:capacity}) serves two purposes. First, it
ensures that riders are accompanied by drivers throughout their trips.
Second, it ensures that vehicle capacities are not exceeded. Constraint
sets (\ref{eq:ident_driver_1})-(\ref{eq:transfer}) collectively
set a limit on the total number of transfers for each rider. Constraint
sets (\ref{eq:ident_driver_1}) and (\ref{eq:ident_driver_2}) register
drivers who contribute to each rider's itinerary (refer
to Proposition 2 in the Appendix). Constraint set (\ref{eq:transfer})
restricts the number of transfers by each rider (refer to Proposition
1 in the Appendix). Finally, all decision variables of the problem
defined in (\ref{eq:var_Xd})-(\ref{eq:Var_U}) are binary variables. 

Although the formulation in model (\ref{equ:determinstic}) is defined
for a ridesharing system in its infancy, it is easy to include additional
terms in the objective function or introduce additional decision variables
and constraint sets to cover different objectives ridesharing systems
may have throughout their lifetime. For instance, we can minimize
the total travel time by riders and drivers by adding terms $\sum_{d\in D}\sum_{r\in R}\sum_{\ell\in L}(t_{j}-t_{i})y_{\ell}^{rd}$
and $\sum_{d\in D}\sum_{\ell\in L_{d}}(t_{j}-t_{i})x_{\ell}^{d}$
to the objective function, respectively. It is possible to minimize/maximize
the number of matched drivers by introducing the decision variable
$z_{d}^{\prime}$ that takes the value of 1 is driver $d$ is matched,
and $0$ otherwise. In this case, the unit values on the right hand
sides of constraint sets (\ref{eq:d_orig}) and (\ref{eq:d_dest})
should be replaced with $z_{d}^{\prime},$ and the term $W^{\prime}\sum_{d\in D}z_{d}^{\prime}$ should be added to the objective function.

In the formulation presented in model (\ref{equ:determinstic}), we
do not consider an exclusive service time for participants. This can
be a realistic assumption, since we are concerned with people and
not goods. However, if service times are required due to practical
considerations, they can be easily integrated by making small changes
in the definition of the link sets and constraint sets (\ref{eq:d_balance})
and (\ref{eq:r_balance}). Originally, we defined a link between two
stations $s_{i}$ and $s_{j}$ if the two stations were neighboring
stations in the network. To account for service times, we have to
redefine the link sets to include links between any two stations.
If a rider travels from their origin station to station $k$, and
then from station $k$ to their destination station, this implies
a transfer at station $k$. To include service times, equations (\ref{eq:d_balance})
and (\ref{eq:r_balance}) should be replaced with constraint sets
(\ref{equ:balance_d_service}) and (\ref{equ:balance_r_serivice}),
respectively, where $t_{s}$ is the service time (in number of time
intervals).

\begin{align}
& \!\!\!\!\!\!\!\!\sum_{_{\ell=(t_{i},s_{i},t,s)\in L}^{\qquad t_{i},s_{i}}}\!\!\!\!\!\!x_{\ell}^{d}\qquad =
 \!\!\!\!\!\!\! \sum_{_{\ell=(t+t_{s},s,t_{j},s_{j})\in L}^{\qquad t_{j},s_{j}}}\!\!\!\!\!\!x_{l}^{d}\quad \forall d\in D,\forall t\in T_{d},\forall s\in S\backslash\{OS_{d}\cup DS_{d}\}\label{equ:balance_d_service}\\
& \sum_{^{d\in D^{\prime}}}
  \!\!\!\!\!\!\sum_{_{\ell=(t_i,s_i,t,s)\in L}^{\qquad \!\!\!\! t_{i},s_{i}:}}\!\!\!\!\!\!y_{\ell}^{rd}=
  \sum_{^{d\in D^{\prime}}}
  \!\!\!\!\!\!\!\!\sum_{_{\ell=(t+t_{s},s,t_{j},s_{j})\in L}^{\qquad t_{j},s_{j}:}}\!\!\!\!\!\!\!\!y_{\ell}^{rd}\quad \forall r\in R,\forall t\in T_{r},\forall s\in S\backslash\{OS_{r}\cup DS_{r}\}\label{equ:balance_r_serivice}
\end{align}

Solving the optimization problem in model (\ref{equ:determinstic})
is computationally prohibitive, even for small instances of the problem.
Therefore, in its original form, the ride-matching problem in model (\ref{equ:determinstic})
is not appropriate for time-sensitive applications. In the next section,
we introduce a pre-processing procedure that reduces the size of the
input sets to the optimization problem. Subsequently, we propose a
decomposition algorithm that attempts to solve the original ride-matching
problem by means of iteratively solving multiple smaller problems called ``sub-problems''.

\section{Pre-processing Procedure}

The goal of the pre-processing procedure is to limit the number of
accessible links for each participant, and identify and eliminate
drivers who cannot be part of a rider's itinerary due to lack
of spatiotemporal compatibility between their trips. At the onset,
it should be noted that this procedure does not limit the search space
of the optimization problem, but only narrows it by cutting down practically
infeasible ranges, and therefore it does not affect the optimality
of the solution.

As explained before, we present a link, $\ell$, as a 4-tuple $(t_{i},s_{i},t_{j},s_{j})$.
Participants can potentially reach any station in the network in any
time interval within their travel time window, making the size of
the set of links, $L$, as large as $O(|T||S|)$, where $|T|$ is
the number of time intervals in the study time horizon, and $|S|$
is the number of stations in the network.

The premise of the pre-processing procedure is that the spatiotemporal
constraints enforced by maximum ride times and travel time windows
of participants limit their access to members of the link set $L$.
We use this information to construct the set of links accessible to
riders and drivers, denoted by $L_{r}$ and $L_{d}$, respectively.

The origin and destination stations, maximum ride times, and travel
time windows of participants can be used to define a region in the
network in the form of an ellipse, inside which participants have
a higher degree of space proximity, i.e., the percentage of accessible
stations within this region is at least as high as the same percentage
within the entire network. We call the region inside and on the circumference
of the ellipse associated with participant $p$ the reduced graph
of the participant, denoted by $G_{p}$ (reduced graph of rider $r$/driver
$d$ is denoted as $G_{r}/G_{d}$). The focal points of the ellipse
are the participant's origin and destination stations.
The length of the major axes between the focal points is the straight
distance between the origin and destination stations, and the transverse
diameter of the ellipse is an upper-bound on the distance that can
be traveled by the participant in $\frac{T_{p}^{TB}}{\Delta t}$ number
of time intervals, within the participant's travel time window. 

We know that for each point on the circumference of an ellipse, sum
of the distances from the two focal points is always constant, and
equal to the transverse diameter of the ellipse. By setting the length
of the transverse diameter to the maximum distance a participant can
travel given their travel time window, we ensure that none of the
stations outside of the participant's reduced graph are accessible
to them.

After the reduced graphs are generated, we use the link reduction
algorithm presented in Algorithm \ref{alg:link_reduction} to construct
sets $L_{r}$ and $L_{d}$. The algorithm finds the set of links for participant $p$ in two steps: a forward movement followed by a backward movement. Both forward and backward movements are iterative procedures. 
We start the algorithm by defining sets $T_s$ and $L(s)$ for all stations $s\in G_p$ and initializing them to be empty, where $T_s$ is the set of time intervals during which station $s$ can be reached, and $L(s)$ is the set of links terminating at station $s$ (line 2).

{\footnotesize{}}
\begin{algorithm}
{\footnotesize{}\caption{Link reduction algorithm\label{alg:link_reduction}\footnotesize{}
}{\footnotesize \par}

\textbf{Generate a link set $L_{p}$ for participant $p$}

$01\quad$\textbf{Initialize}

$02\quad$ $T_s=\emptyset, \;L(s)=\emptyset,\; \forall s\in G_{p}$

$03\quad$ \textbf{Step1. Forward movement}

$04\quad$ $S_{act}=\{OS_{p}\}$

$05\quad$ $\bar{S}_{act}=\emptyset$

$06\quad$ $T_{OS_p}=\{\frac{T^{ED}_p}{\Delta t},\frac{T^{ED}_p}{\Delta t}+1, \frac{T^{ED}_p}{\Delta t}+2,..., \frac{T_{p}^{LA}-T_{static}(OS_{p},DS_{p})}{\Delta t}\}$

$07\quad$ While $S_{act}\not=\emptyset$

$08\quad$ $\qquad s_{1}\leftarrow S_{act}(1)$

$09\quad$ $\qquad$For $s\in\{S\backslash OS_p\}$

$10\quad$ $\qquad\qquad$ Set $T_{s_1} = T_{s_1} \cup \{t_1\}$ such that $\exists \: \ell = (t,s,t_1,s_1)\in L(s), \forall (t,t_1)\in T_p^2$ 

$11\quad$ $\qquad$End For

$12\quad$ $\qquad S_{act}=S_{act}\backslash\{s_{1}\}$

$13\quad$ $\qquad\bar{S}_{act}=\bar{S}_{act}\cup\{s_{1}\}$

$14\quad$ $\qquad$For $s_{2}\in S:(s_{1},s_{2})\in G_{p}$

$15\quad$ $\qquad\qquad$If $s_{2}\cap\bar{S}_{act}=\emptyset$

$16\quad$ $\qquad\qquad\qquad S_{act}=S_{act}\cup\{s_{2}\}$

$17\quad$ $\qquad\qquad\qquad$For $t\in T_{s_{1}}$

$18\quad$ $\qquad\qquad\qquad$$\qquad L(s_{2})= L(s_{2})\cup \{(t,s_{1},t + T_{dynamic}(t,s_{1},s_{2})^{\dagger}$$,s_{2})\}$

$19\quad$ $\qquad\qquad\qquad$End For

$20\quad$ $\qquad\qquad$Else

$21\quad$ $\qquad\qquad\qquad$For $t\in T_{s_{1}}$

$22\quad$ $\qquad\qquad\qquad\qquad$If $T_{s_{2}}\cap\big\{ t+T_{dynamic}(t,s_{1},s_{2})\big\}\neq\emptyset$

$23\quad$ $\qquad\qquad\qquad\qquad\qquad L(s_{2})=L(s_{2})\cup\{(t,s_{1},t + T_{dynamic^{\dagger}}(t,s_{1},s_{2}),s_{2})\}$

$24\quad$ $\qquad\qquad\qquad\qquad$End If

$25\quad$ $\qquad\qquad\qquad$End For

$26\quad$ $\qquad\qquad$End If

$27\quad$ $\qquad$End For

$28\quad$ End While

$29\quad$ \textbf{Step 2. Backward movement}

$30\quad$ $L_{del}=\{(t_{1},s_{1},t_{2},s_{2})\in L_{p}:(s_{2}=DS_{p})\wedge (t_{2}>\frac{T_{p}^{LA}}{\Delta t})\}$

$31\quad$ While $L_{del}\not=\emptyset$

$32\quad$ $\qquad\ell(t_{1},s_{1},t_{2},s_{2})\leftarrow L_{del}(1)$

$33\quad$ $\qquad L_{del}=L_{del}\backslash\{\ell\}$

$34\quad$ $\qquad L_{p}(s_{2})=L_{p}(s_{2})\backslash\{\ell\}$

$35\quad$ $\qquad$For $s_{1}\in G_{p}$

$36\quad$ $\qquad\quad\;\;$For $(t,s):(t,s,t_{1},s_{1})\in L_{p}(s_{1})$ 

$37\quad$ $\qquad\qquad\quad L_{del}=L_{del}\cup\{(t,s,t_{1},s_{1})\}$

$38\quad$ $\qquad\quad\;\;$End For

$39\quad$ $\qquad$End For

$40\quad$ End While

$41\quad$ \textbf{Generating Link set $L_p$}

$42\quad$ $L_p = \bigcup_{s \in G_p} L(s)$

$43\quad$ $^{\dagger}T_{dynamic}(t,s_{1},s_{2}):$ travel time between stations
$s_{1}$ and $s_{2}$ at time interval $t$ 
}
\end{algorithm}
{\footnotesize \par}

In each iteration the forward movement generates a set of links originating from one of the stations in the reduced graph of participant $p$. We start the forward movement by defining the set of active stations, $S_{act}$, and initializing this set with the origin station of participant $p$ (line 4). 
The time intervals during which this origin station can be reached can be easily computed using the equation in line 6, where $T_{static}(s_1,s_2)$ denotes the shortest path travel time between stations $s_1$ and $s_2$. Set $T_{OS_p}$ in line 6 maintains indices of time intervals during individual $p$'s travel time window. To ensure that no feasible links are eliminated, $T_{static}$ should contain underestimated link travel times, say, for example, the travel times during non-peak hours.  
The algorithm then selects the first member of the active stations set, denoted by $s_{1}$ (line 8), and updates the set of time intervals during which $s_1$ can be reached by examining all the previously generated links that terminate at $s_1$ (lines 9-11). Once station $s_1$ has been processed (i.e., the time intervals during which this station can be reached are obtained), this station will be removed from set $S_{act}$ (line 12), and added to set $\bar{S}_{act}$ (line 13), which maintains a list of previously processed stations. Note that adding station $s_1$ to set $\bar{S}_{act}$ does not mean that the set $T_{s_1}$ is finalized; we might need to revise this set later, as the participant may need to visit a station more than once if their travel time window allows (e.g., a driver may visit a station twice to pick up different passengers that start their trips at different time intervals).

Next, outgoing links from $s_1$ whose end stations are inside
the reduced graph are identified (line 14) and their end stations are added
to the set of active stations (line 16). At this point, we have information on the starting station $(s_{1})$, ending stations ($s_{2}:(s_{1},s_{2})\in G_{p})$ and the starting time intervals at $s_{1}$ $(T_{s_1})$. The ending time intervals for each $s_{2}$ can be easily looked up from a dynamic (i.e., time-dependent) travel time matrix, $T_{dynamic}$, completing the information required to construct the set of links originating at $s_{1}$ (lines 17-19). Note that once we identify station $s_2$ on the reduced graph, the process of generating $L(s_2)$ will be slightly different depending on whether $s_2$ is a member of set $\bar{S}_{act}$ (lines 17-19) or not (lines 21-25), ensuring that a single link is not added to the link set multiple times. We iterate the forward movement until the set of active stations becomes empty.

Figure \ref{fig:Ellipse} demonstrates an example of the forward movement
for a participant who is traveling from station $14$ to station $8$,
with $TW_{p}=[1,40](\Delta t=1$ min$)$, maximum ride time of $40$
minutes, and shortest path travel time of $38$ minutes. Link travel
times (in minutes) are shown on the graph. It is assumed that the
travel time remains constant on each link. Note that this assumption
is made only for simplicity, and using time-dependent travel times
would be just as straightforward. The set of links for each station
is computed during the forward movement, and is presented in the Figure
\ref{fig:Ellipse}. These links are not final, however, and have to
be refined during the backward movement.

The backward movement simply scans through the set of links generated
for each station by the forward movement and refines these sets by
removing the time intervals that are identified as infeasible based
on the participant's latest arrival time. 
The backward movement starts by identifying links that take the participant to his/her destination station after the participant's latest arrival time, and adds these links to set $L_{del}$, which is initially defined as an empty set (line 30). The algorithm then goes through the members of set $L_{del}$ one by one. For each member $\ell=(t_1,s_1,t_2,s_2)$, it removes $\ell$ from set $L_p(s_2)$ (line 34), identifies links with end nodes $(t_1,s_1)$ (line 36), and adds these links to set $L_{del}$  (line 37). The backward movement ends when set $L_{del}$ becomes empty. At this point, the union of all the remaining links $L(s), \forall s\in G_p$ yields the set of links for participant $p$ (line 42).

In the example shown in Figure \ref{fig:Ellipse}, since the latest arrival time is at $\Delta t=40$, links with ending time intervals $41$ and $42$ should be removed from the set of links for the destination station. Tracking back the
stations from destination to origin, the time intervals for the stations
that have led to the infeasible time intervals at the destination
station are identified and removed (see Algorithm \ref{alg:link_reduction}
for details). For the example in Figure \ref{fig:Ellipse}, after
completing the backward movement, the set of links, $L_{p}$, is derived
and listed in the figure.

\begin{figure}
\begin{centering}
\includegraphics[width=5in]{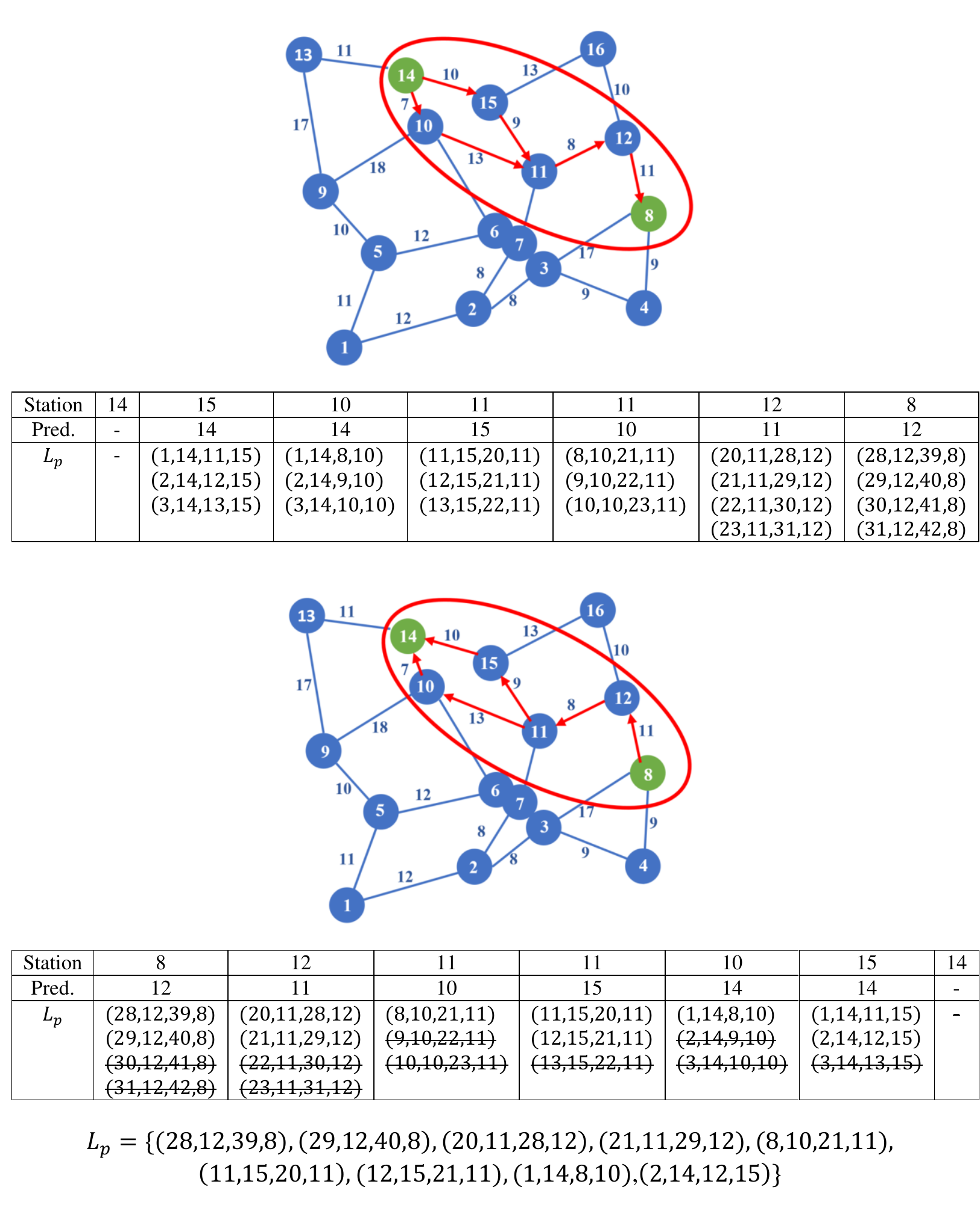}
\par\end{centering}

\caption{An example of forward and backward movements for a participant with
$OS_{p}=14$, $DS_{p}=3$, $\Delta t=1$ min, $TW_{p}=[1,40]$, and
\label{fig:Ellipse}$T_{p}^{TB}=40$ min }
\end{figure}

Once the set of links for participants are generated, we can use these
sets to reduce the size of some other sets in the optimization problem,
namely the set of riders and drivers. We can reduce the size of the
rider set $R$ by filtering riders out of the problem based on their
accessibility to potential drivers. For a rider to be served, they
should have spatiotemporal proximity with at least one driver at both
their origin and destination stations. Riders of set $R$ who do not
enjoy this spatiotemporal proximity can be filtered out. 

In addition, in order for a driver $d$ to be able to contribute to
the itinerary of a rider $r$, the intersection of link sets of the
two should not be empty, i.e., $L_{r}\bigcap L_{d}\neq\emptyset$.
We denote by $M$ the set of tuples $(r,d)\in R\times D$ for whom
$L_{r}\bigcap L_{d}\neq\emptyset$ and therefore could potentially
be matched. For members of set $M$, we construct a set $L_{rd}=\{L_{r}\bigcap L_{d}\}$.
Furthermore, we add tuples $(r,d^{\prime})$ to set $M$, and set
$L_{rd^{\prime}}=L^{\prime}$ for all the unfiltered riders $r\in R$.

Note that the ellipses that form boundaries of the reduced graphs
strictly prevent any feasible links from being cut off, and therefore
the link set $L$ can be safely replaced by $L_{p},\forall p\in P$.
Furthermore, since the potential drivers for each rider are determined
based on the reduced graphs, no potential driver is excluded from
the set of eligible drivers for each rider. Hence, using the refined
input sets generated by the pre-processing procedure in the ride-matching
problem does not affect the optimal matching of riders and drivers.

\section{Decomposition Algorithm}

The decomposition algorithm attempts to solve the original ride-matching
problem via iteratively solving a number of smaller problems called ``sub-problems'' that are easier to solve. The algorithm pseudo-code is described in Algorithm 2, in Appendix \ref{sec:Decomposition-Algorithm}. The
basic idea is that in each iteration the algorithm solves a number
of sub-problems that can represent the entire system. If the solutions
to these sub-problems do not have any conflicts, the algorithm is terminated
and the union of solutions to the sub-problems yields the global optimum
for the original problem (proof in section \ref{sub:Optimality}).
The algorithm flowchart is displayed in Figure \ref{fig:The-decomposition-algorithm}.

\begin{figure}
\begin{centering}
\includegraphics[width=3in]{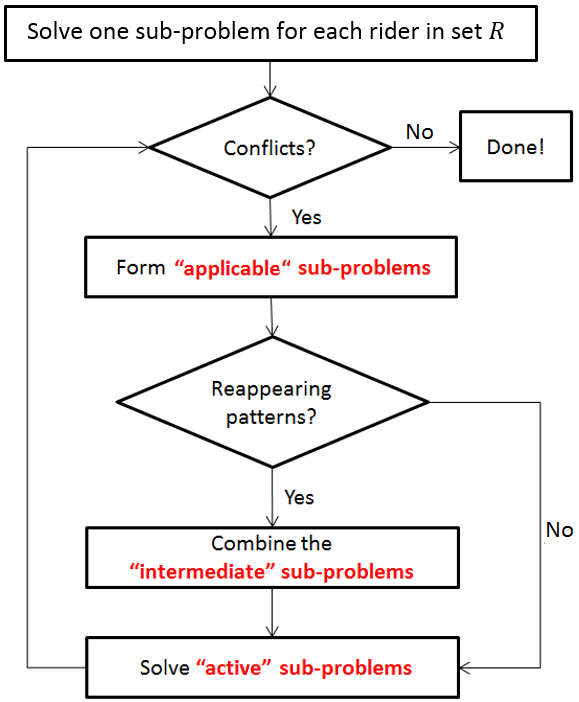}
\par\end{centering}
\caption{The decomposition algorithm flowchart\label{fig:The-decomposition-algorithm}}
\end{figure}

Let $R_{i}^{k}$ and $D_{i}^{k}$ denote the set of riders and drivers
in sub-problem $k$ of iteration $i$, respectively. Each sub-problem
includes a subset of riders in the problem. Once the subset of riders
for sub-problem $k$ in iteration $i$ is determined, the set of drivers
can be formed as $D_{i}^{k}=\{d|\forall r\in R_{i}^{k},(r,d)\in M\}$.

The algorithm starts by solving $|R|$ sub-problems, each including
one of the riders in set $R$. In the case of there being no conflicts
between the solutions, the solution to the original problem is readily
available. This happens if each rider is matched with a different
driver, or if multiple riders are matched with the same driver and
the driver is capable of performing all the pick-up and drop-off assignments
for the assigned riders. If not, conflicts are identified. Note that
existence of conflicts between drivers' paths in different sub-problems
implies that the union of solutions to the sub-problems is infeasible
to the original ride-matching problem. 

In each iteration, in the case of there being conflicts between solutions of the sub-problems in the previous iteration, we form the set of ``applicable'' sub-problems. An ``applicable'' sub-problem is comprised of: $(i)$ sub-problems from the previous iteration from which riders have been excluded (with the remaining set of riders), and $(ii)$ a group of riders from the last iteration's sub-problems with identical driver assignment (if these assignments conflict in time or space). Note that in the latter case if there are multiple groups of riders with identical but conflicting driver assignments, any such groups can be selected to form the new sub-problem. Sub-problems in the previous iteration that are not applicable sub-problems will be carried out to the next iteration without any change.

After the new set of sub-problems are formed, first we have to check
to see if there are any loops between iterations. If the set of sub-problems
in the current iteration is similar to the set of sub-problems in
a previous iteration, the algorithm will be looping between iterations
if no measures are taken. To prevent this, we re-define sub-problems
in the current iteration by forming an ``intermediate''
sub-problem whenever a loop is identified. An ``intermediate''
sub-problem combines the sub-problems from the previous iteration
that contribute to the loop, and hence prevent it (refer to Algorithm
2 in Appendix \ref{sec:Decomposition-Algorithm} for details.)

After all the new sub-problems are determined, a decision has to be
made on whether a sub-problem needs to be solved or not. Sub-problems
that need to be solved are called ``active'' sub-problems. These sub-problems are the ones whose optimal solutions
cannot be readily obtained from the previous solutions. Sub-problems
that have already been solved in the previous iterations (such as
non-applicable sub-problems) are not active. In addition, if a sub-problem
is the union of multiple sub-problems in a previous iteration, and
the solutions of these sub-problems do not conflict, the solution
to the sub-problem can be readily obtained by combining the solutions
of the non-conflicting sub-problems. The algorithm stops if the solutions
to the current iteration's sub-problems do not have any conflicts,
i.e., each driver is assigned to one route only. 

Note that using this decomposition algorithm, a large problem could
be solved in the first iteration, or we might end up solving multiple
sub-problems before having to solve the original problem in the last
iteration. The main merit of this algorithm is that sub-problems in
each iteration can be solved independently. This also indicates that
parallel computing implementations are possible. 

After performing the pre-processing procedure and while applying
the decomposition algorithm, we use a refined version of the the optimization
problem in model (\ref{equ:determinstic-1}), presented in Appendix
\ref{sec:Revised-version-optimization}. This refined problem (model
\ref{equ:determinstic-1}) is very similar to model \ref{equ:determinstic},
with two major differences: (i) Constraint sets (\ref{eq:d_travel_time})
and (\ref{eq:r_travel_time}) are now redundant, since the requirement
to not exceed the maximum ride time is met when we form the reduced
graphs and perform the forward and backward movements in the pre-processing
procedure, and (ii) input sets in model (\ref{equ:determinstic-1})
are more refined owing to both the pre-processing procedure and the
decomposition algorithm.

\subsection{Illustrative Example\label{sub:Illustrative-Example}}

Assume that a ridesharing system has 6 riders and 4 drivers, and that
all drivers have spatiotemporal proximity with all riders, i.e., $\forall\: (r,d)\in R\times D,\;(r,d)\in M$.
The iterations of the decomposition algorithm are displayed in Figure
\ref{fig:exp1}. As the interest is in showing the nature of the creation
of sub-problems, the actual network on which this problem is solved
is left out. The active sub-problems during each iteration are displayed
in blue in the figure. The sub-problems whose solutions do not change
throughout the iterations of the algorithm are displayed in green.

In iteration 1, each rider constitutes an active sub-problem. The
solutions show that riders 1, 3, and 6 all have driver 1 in their solution,
but in conflicting paths. Therefore, an active sub-problem of $\{r_{1},r_{3},r_{6}\}$ is formed and solved in the second iteration. Also, since rider 2 was not able to find any matches even without facing competition from other riders, he/she will not be able to find a match in the current
configuration of the system. Sub-problems $\{r_{4}\}$ and $\{r_{5}\}$
are not active sub-problems and their solutions are readily available.

The solution to the active sub-problem $\{r_{1},r_{3},r_{6}\}$ in
iteration 2 indicates that the optimal matches for riders 5 and 6
are in conflict (they are both matched with driver 2, but through
different paths.) So they form the applicable sub-problem $\{r_{5},r_{6}\}$.
Also, since rider 6 was removed from the sub-problem $\{r_{1},r_{3},r_{6}\}$,
the solution obtained for this sub-problem for riders 1 and 3 might
not be optimal anymore. Therefore, a new applicable sub-problem $\{r_{1},r_{3}\}$ is formed. However, not both of these newly formed sub-problems are active. The optimal match for rider 5 is driver 2, and the optimal match for rider 6 is driver 1. Since these two do not have any conflicts, the solution to the $\{r_{5},r_{6}\}$ sub-problem is readily available.

The only active sub-problem in iteration 3 is $\{r_{1},r_{3}\}$.
The solution to this sub-problem suggests that two new applicable
sub-problems $\{r_{1},r_{3},r_{6}\}$ and $\{r_{5}\}$ need to be
formed, which along with two sub-problems $\{r_{4}\}$ and $\{r_{2}\}$
should constitute the set of sub-problems for iteration 4. However,
we had the exact same set of sub-problems iteration 2. Therefore,
in order to avoid looping, a new (intermediate) sub-problem $\{r_{1},r_{3},r_{6},r_{5}\}$
is formed in iteration 4. After solving this sub-problem, we find
that there are no more conflicts. Hence, the global solution to the original problem is obtained in iteration 4. A total of 9 sub-problems needed
to be solved for this solution to be obtained. However, in iteration
1, all 6 active sub-problems could be solved simultaneously.

\begin{figure}
\begin{centering}
\includegraphics[width=4in]{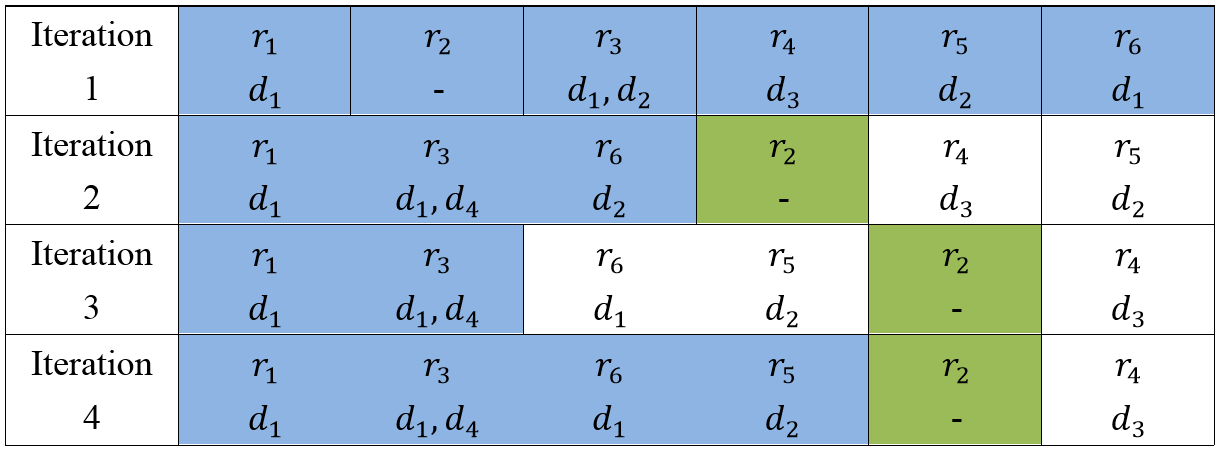}
\par\end{centering}
\caption{Iterations of the decomposition algorithm\label{fig:exp1}}
\end{figure}

It is possible to use a simpler version of the algorithm which is
easier to implement, but may take longer to solve. In this simplified
version, if any two riders in two sub-problems have conflicts, all
the riders in the two sub-problems are combined into a new sub-problem
in the following iteration. This will lead to potentially fewer iterations, but larger sub-problems to be solved in each iteration. 

Let us apply this simplified algorithm to the example above. The sub-problems
in each iteration are presented in Figure \ref{fig:exp2}. Here, after
solving the active sub-problem $\{r_{1},r_{3},r_{6}\}$ in iteration
2, and studying the solutions of all sub-problems, it turns out that
riders 6 and 5 are both matched with driver 2, but through conflicting
paths. Therefore, in the next iteration the two sub-problems $\{r_{1},r_{3},r_{6}\}$ and $\{r_{5}\}$ are combined. In this particular example, using the simplified version of the algorithm leads to reaching the optimal solution in fewer iterations, and less amount of time, since in Figure
\ref{fig:exp2} we are skipping iteration 3 in Figure \ref{fig:exp1}.
Although reaching the optimal solution in fewer iterations is expected
using the simplified version of the algorithm, a smaller solution
time is not a typical behavior to be expected from the simplified
version of the algorithm.

\begin{figure}
\begin{centering}
\includegraphics[width=4in]{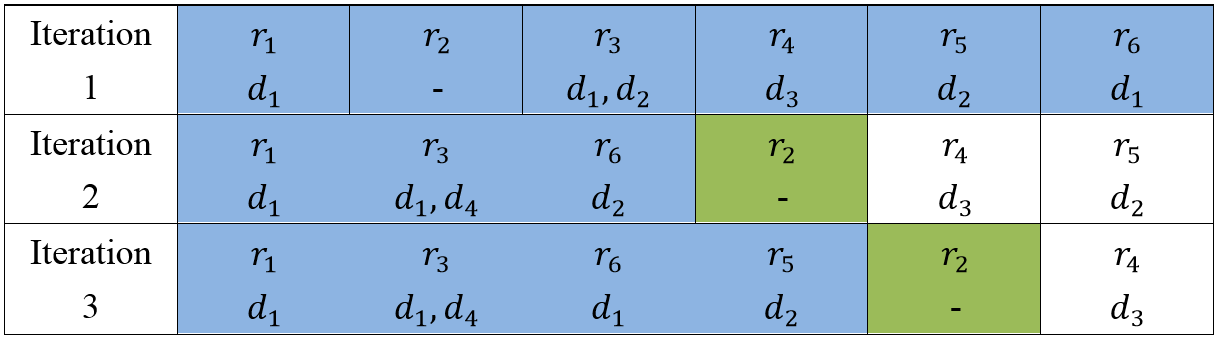}
\par\end{centering}
\caption{Iterations of the simplified decomposition algorithm\label{fig:exp2}}
\end{figure}

\subsection{Properties of the Decomposition Algorithm}

\subsubsection{Optimality\label{sub:Optimality}}

During each iteration of the decomposition algorithm, we are in fact
solving a relaxation of the original ride-matching problem. In the
original optimization problem each driver can be assigned to a single
route. By including a driver in multiple sub-problems, we might have
the driver being routed differently in each. In fact, we are relaxing
the constraint sets that force each driver to take a single route.
Throughout the iterations, we are trying to find conflicts of this
nature, and merge the sub-problems whose solutions demonstrate such
a conflict, in an attempt to find a solution with no conflicts.

If the optimal solution to this relaxation satisfies the omitted set
of constraints, then this solution would be optimal to the original
problem as well. We have based the stopping criterion of our algorithm
on this factual principle. After the final iteration of the algorithm,
there exist no more conflicts between any given driver's
routes in different sub-problems. Therefore, what we find through
the use of the decomposition algorithm is indeed an optimal solution
to the relaxation of the original problem that does not violate the
set of omitted constraints. Hence the solution found is the globally
optimal solution to the original problem. 

It should be noted that throughout the iterations, the union of the
sub-problem solutions form an infeasible solution to the original
problem. It is only in the last iteration, where there are no conflicts
among driver assignments, when the first feasible solution is obtained,
and this feasible solution is in fact the optimal solution. Furthermore,
note that although the pre-processing procedure may eliminate some
of the drivers from the pool of drivers available for each rider,
it does not affect the optimality of the final matching. The omitted
drivers could not have contributed to the solution in any case, since
they did not have spatiotemporal proximity to the riders. Additionally, note that the proposed algorithm is optimal as long as maximizing the matching rate (or any rider-related terms, such as minimizing the total VMT by riders) remains the primary objective of the system. Otherwise, the decomposition algorithm provides a bound on the optimal solution.

\subsubsection{Bounds\label{sub:Bounds}}

In this section we discuss computing upper and lower bounds on the
optimal value of the optimization problem in model (\ref{equ:determinstic})
after each iteration of the decomposition algorithm. We compute the
bounds for the objective function in (\ref{eq:Obj}), but the concept
can be easily extended to a variety of objective functions. 

Clearly, if a ride-matching problem is solved to optimality, the upper
and lower bounds would be similar. However, if there is a time limit
on reaching a solution, we might be obliged to settle for a sub-optimal
solution (as will be described later in section \ref{sub:Heuristic-solutions}).
In that case, upper and lower bounds on the optimal solution can be
computed after completion of each iteration to provide insight on
the quality of the solution.

The main objective in model (\ref{equ:determinstic}) is to find the
maximum number of served riders, and the second term in (\ref{eq:Obj})
has been added only for technical reasons (to ensure that the constraint
sets work properly, as explained in Proposition 1 in Appendix \ref{sec:Proof-of-Proposition}.)
Therefore, while solving the sub-problems, we include this term, but
use only the value of the first term as the value of the objective
function of the problem while deriving bounds. 

In section \ref{sub:Optimality}, we discussed that the solution reached
in every iteration of the decomposition algorithm is infeasible, until
the last iteration. This infeasibility is caused by conflicts between
itineraries of common drivers in different sub-problems. When such
conflicts exist, the best case scenario is for all the riders who
are receiving conflicting itineraries to be eventually served in the
system. Therefore, the upper bound in each iteration can be computed
as the sum of the number of served riders in all sub-problems, regardless
of the conflicting driver itineraries. It should be noted that the
upper bound is strictly non-increasing with the number of iterations.
The reason is that when two sub-problems have conflicting assignments
during a given iteration, the riders who are receiving conflicting
itineraries will be included in the same sub-problem in the following
iteration. If all such conflicting riders can be served, then the
upper bound to the objective function does not change. Otherwise,
the upper bound will decrease.

The lower bound in each iteration can be obtained by solving a set
packing problem. The optimal solution to sub-problems of a typical
iteration provides us with information on the riders with conflicting
itineraries, and drivers who form the itineraries of such riders.
To find the tightest lower bound, we have to find the maximum number
of riders who could be served under the current configuration (i.e.,
assuming that the driver routes are fixed.) This is analogous to solving
a set packing problem where the universe is the set of drivers, and
the sub-sets are drivers who form each rider's itinerary.
It should be noted that unlike the upper bound that has a non-increasing
trend with the number of iterations, the lower bound is not necessarily
non-decreasing with iterations. In addition, note that the lower bound
corresponds to a feasible solution, while the upper bound corresponds
to a possibly infeasible solution.
Figure \ref{fig:Bounds} shows the lower and upper bounds for the example in section \ref{sub:Illustrative-Example}. Clearly the lower
and upper bounds meet at the final iteration.

\begin{figure}
\begin{centering}
\includegraphics[width=3in]{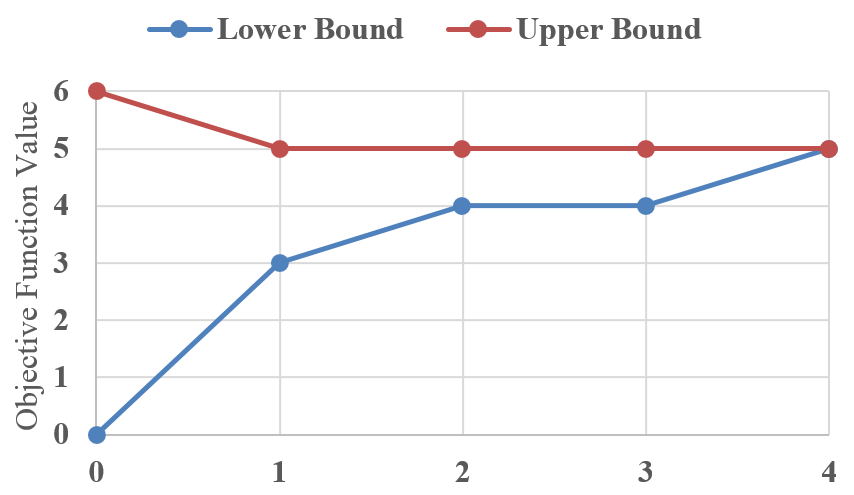}
\par\end{centering}

\caption{Bounds on the objective function value for the example in section
\ref{sub:Illustrative-Example}\label{fig:Bounds}}
\end{figure}

As we mentioned in the beginning of this section, we can follow a
similar logic to compute bounds for a variety of objective functions, as long as serving the maximum number of riders remains the primary objective of the system. Since we are solving a relaxation of the original ride-matching problem
at each iteration, no matter the objective function, sum of the objective
functions of sub-problems always provides an upper bound on the original
problem. To obtain a lower bound at each iteration, similar to what
we did for the objective function in (\ref{eq:Obj-1}), we have to
derive a feasible global assignment from the solutions of sub-problems
at each iteration by solving a set packing problem. At each iteration,
once we find the set of riders who can be served, the drivers who
serve them and the itineraries of these riders and drivers are readily
available, and can be used to calculate any measures that are included
in the objective function to obtain a lower bound.

\section{Numerical Study}\label{sec:numericalstudy}

To evaluate the performance of the proposed decomposition algorithm,
we generate and solve 420 random instances of the ride-matching problem.
Each problem instance has a different size in terms of number of riders
and drivers, and is generated in a grid networks with 49 stations.
The results reported here are averaged over 10 runs for each problem instance. 

The number of participants $|P|$ in the problem instances varies
between 20 and 400, and the number of riders $|R|$ between 1 and
$|P|$. Origins and destinations of trips are selected uniformly at
random among the 49 stations. Earliest departure times of all trips
are generated randomly within a one hour time period. The maximum
ride time for a participant $p$ is generated randomly based on a ``travel time budget factor'' of 1.1, indicating that the particpant's maximum ride time takes a value within the window $[tt_{(OS_{p},DS_{p})},1.1\thinspace tt_{(OS_{p},DS_{p})}]$, where $tt_{s_1,s_2}$ is the travel time between stations $s_1$ and $s_2$.
The latest arrival times are simply the sum of the earliest departure
times and the maximum ride times of participants. Each vehicle is
assumed to have 4 empty seats, and each rider is assumed to accept
up to three transfers. 

The ride-matching instances are solved on a PC with Core i7 3 GHz
and 8GB of RAM. The optimization problems are coded in AMPL, and solved
using CPLEX 12.6.00 with standard tuning. 

In addition to solving the 420 problem instances with the stated default parameter values, we solve additional problem instances in the following sections to conduct sensitivity analysis over some of the parameters of the
problem, including the number of stations, the spatiotemporal proximity
of trips, and the maximum ride times of participants. The parameter values used in the rest of the paper are the default values, unless otherwise specified. Furthermore, the results presented in the rest of the paper are for solving the static versions of problems, unless otherwise specified.

\subsection{Pre-processing\label{sub:Pre-processing}}

In this section we look at the running time of the pre-processing
procedure, and examine its impact on the size of the ride-matching
problem that needs to be solved. Figure \ref{fig:Pre_Proc_time} shows
the contour plot of the pre-processing time as a function of the size
of the problem. The pre-processing time is the sum of two components:
the time required to generate the set of feasible links for each participant,
and the time required to find the set of feasible drivers for each
rider. The former component consists of the time spent on generating
reduced graphs, and forming sets of feasible links for participants
using forward and backward movements. This task can be completed when
a participant registers a request in the system, and hence is not
a bottleneck when it comes to solving the ride-matching problem in
real-time in a rolling time horizon implementation. Even if a large
number of participants join the system at the same time, the computations
can be done independently, and implemented in a parallel computing
system. 

The latter component of the pre-processing time consists of the time
required to compare each rider's set of feasible links
to those of the drivers. Similar to the first component, these computations
can be done once a rider joins the system (with all registered drivers),
and updated continuously as drivers keep joining the system.

The total pre-processing time depends on the size of the problem.
For the 420 problem instances solved in this section, the maximum
time was 3.5 seconds. The distribution of the pre-processing times
over the range of sizes of the randomly generated instances are displayed
in Figure \ref{fig:Pre_Proc_time}. The pre-processing procedure managed
to reduce the average size of the link sets for participants to $0.01
$ of the size of the original link set $L$.

Not only does the pre-processing procedure lead to indirect savings
in solution times by limiting the size of the input sets to the optimization
problem, but more importantly a quick scan of the participants'
feasible links can lead to direct and more considerable savings in
solution times. For a rider to be served, he/she should have spatiotemporal
proximity with at least one driver at both origin and destination stations.
Whether such proximity exists can be easily examined using the sets
of feasible links. Riders who do not enjoy spatiotemporal proximity
at both origin and destination stations can be filtered out from the
system. Figure \ref{fig:Pre_Proc_Filter} shows the percentage of riders
filtered out from each problem instance during the pre-processing
procedure. This figure suggests that, for a given number of riders,
the lower the number of drivers, the higher the number of filtered
riders, as expected. 

\begin{figure}
\subfloat[Pre-processing time (sec)\label{fig:Pre_Proc_time}]{\begin{centering}
\includegraphics[width=3.3in]{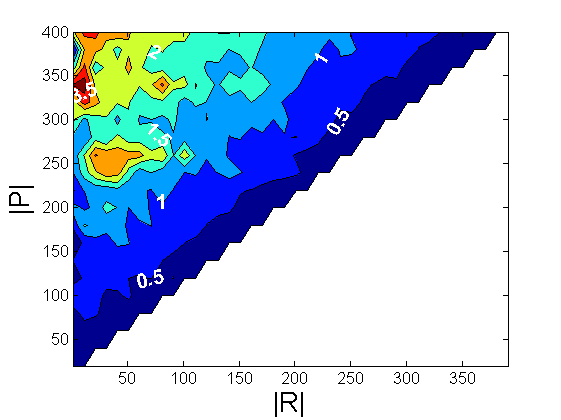}
\par\end{centering}
}\hfill{}\subfloat[Percentage of filtered riders\label{fig:Pre_Proc_Filter}]{\begin{centering}
\includegraphics[width=3.1in]{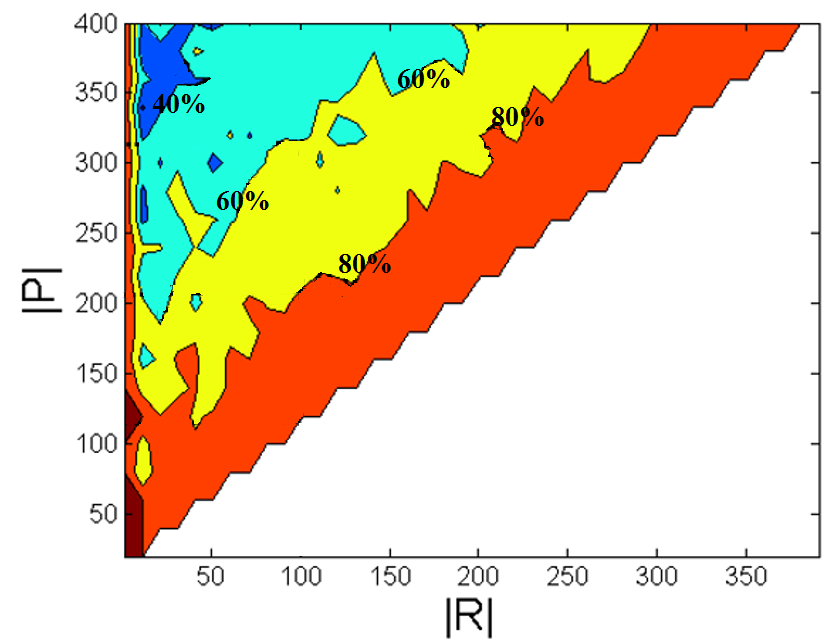}
\par\end{centering}
}
\caption{The pre-processing procedure\label{fig:Pre_Proc}}
\end{figure}

\subsection{Value of a Multi-hop Solution\label{sub:Value-of_Multihop}}

A ridesharing system can use a variety of approaches to match riders
with drivers. Ride-matching methods could differ in multiple aspects,
ranging from the flexibility in routing of drivers, to the type of
routes devised for riders (single- or multi-hop). Matching methods
can be formulated as optimization problems, and solved to optimality.
However, different ride-matching methods lead to different system
performance levels. 

In this section, in order to investigate the effectiveness of the
ridesharing system defined in this study, we compare its performance
with a few more common ride-matching methods. Figure \ref{fig:Multi-Hop}
compares the cumulative number of served riders in the 420 randomly-generated
problem instances using five different matching methods. The problem
instances in this figure are sorted by the number of participants.

The first matching method labeled as the ``OD-based'' in Figure
\ref{fig:Multi-Hop} matches riders and drivers who share the same
origin and destination locations, if their travel time windows allow. The
second matching method, labeled as ``Single-hop, Fixed-route'' matches
each rider with a single driver, where the drivers' routes are pre-specified
and fixed (although they still have flexibility in departure time.)
The third matching method labeled as ``Multi-hop, Fixed-route''
allows riders to transfer between drivers. The drivers' routes, however,
remain pre-specified. The fourth matching method labeled as ``Single-hop,
Flexible-route'' matches each rider with one driver only. In this
method, however, the ridesharing system takes over routing of the
drivers. The fifth and final matching method labeled as ``Multi-hop,
Flexible-route'' allows riders to transfer between drivers. The driver
routes are not pre-specified, and will be determined by the system. 

As Figure \ref{fig:Multi-Hop} indicates, the ``OD-based'' matching
method provides the least number of matches. The ``Single-hop, Fixed-route''
method does significantly better, since using this method not only can
drivers carry the riders who share the same origin and destination
locations with them (as in the ``OD-based'' method), but also they
can carry passengers whose origin and destination locations lie on
their pre-specified routes. The ``Multi-hop, Fixed-route'' matching
method produces results that are slightly superior to the ``Single-hop,
Fixed-route'' method, implying that allowing riders to transfer between
drivers is in general beneficial to the system. 

The ``Single-hop, Flexible-route'' and ``Multi-hop, Flexible-route''
methods lead to considerable improvements in the number of matches.
This implies that having the system route drivers can be one of
the most influential features of a ridesharing system. Comparison
of the ``Single-hop, Flexible-route'' and ``Multi-hop, Flexible-route''
methods suggests that allowing transfers in the system is the second
most important factor, after routing drivers. In addition, comparing
the ``Single-hop, Flexible-route'' and ``Multi-hop, Flexible-route''
methods with ``Single-hop, Fixed-route'' and ``Multi-hop, Fixed-route''
methods suggests that the value of a multi-hop solution becomes more
prominent when the driver routes are not pre-specified. 

\begin{figure}
\begin{centering}
\includegraphics[width=3.5in]{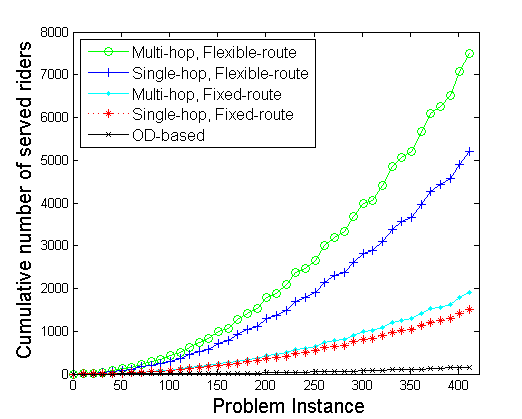}
\par\end{centering}
\centering{}\caption{Cumulative number of served riders using five different matching algorithms\label{fig:Multi-Hop}}
\end{figure}

Figure \ref{fig:Multi-Hop} measures the efficiency of the matching
methods based on the number of served riders. What this figure fails
to demonstrate is how measures of quality of service, such as number
of transfers and waiting times in transfer stations for riders, and
the average vehicle occupancy and the extra time spent in the network
by drivers compare between these methods. Table \ref{tab:Quality-of-service-multihop}
displays these measures of quality of service along with some additional
metrics to assess the five matching methods from different perspectives. 

The results presented in this table are averaged over 10 runs with
400 participants, 200 of which are drivers and 200 are riders. This
problem configuration is selected because it yields the highest number
of served riders, regardless of the matching method. The earliest
departure times of all trips are randomly generated within a 60 minute
time period, and the trip origin and destination stations are selected
randomly among the 49 stations. The quality of service measures in
this table are averaged over all participants in all randomly-generated
problems. 

Table \ref{tab:Quality-of-service-multihop} suggests that in addition
to serving higher number of riders, the multi-hop matching methods
also engage higher number of drivers, compared to their single-hop
counterparts. Although the ultimate purpose of a ridesharing system
is serving riders, involving higher number of drivers is important
as well, because users who do not get matched might stop registering
in the system after a few failed attempts. Another interesting point
is that in multi-hop matching methods the number of matched drivers
is higher than the number of served riders. In spite of this observation,
the average number of riders a driver in a multi-hop system carries
is slightly higher than the average number of riders carried by a
driver in a single-hop system. 

Another interesting observation is that vehicle occupancies are highest
in a multi-hop fixed-route system. In such a system, because driver's
routes are fixed, a smaller percentage of drivers get matched. However,
drivers who get matched are the ones who drive on ``popular'' routes,
and therefore end up with higher occupancy rates. In other words,
although the system is not routing drivers, those drivers who try to
optimally route themselves will benefit more compared to the case where
the system routes all drivers.

Table \ref{tab:Quality-of-service-multihop} suggests that the measures
of quality of service stay within reasonable ranges for all matching
methods. It should be noted, however, that these measures correlate
with the problem parameters in general, and are more sensitive to
some parameters than others, as we will see in the following sections. 

\begin{landscape}
\begin{table}
\footnotesize
\caption{Quality of service measures for the five matching methods in a system
with relatively low spatiotemporal proximity among trips\label{tab:Quality-of-service-multihop}}
\tabcolsep .22pc
\centering{}%
\begin{tabular}{llccccccc}
\hline 
 &  & \multicolumn{3}{c}{Riders} &  &  & \multicolumn{2}{c}{Drivers}\tabularnewline
\cline{3-5} \cline{7-9} 
 &  & Num. (\%) & Min\textbackslash{}avg.\textbackslash{}max  & Min\textbackslash{}avg.\textbackslash{}max  &  & Num. (\%) & Min\textbackslash{}avg.\textbackslash{}max  & Min\textbackslash{}avg.\textbackslash{}max \tabularnewline
Matching method &  & served & \ num. transfers & \ wait in transfer (min) &  & involved & \ extra travel time (min) & riders on board\tabularnewline
\hline 
\hline 
OD-based &  & 5 (3\%) & NA & NA &  & 5 (3\%) & NA & 1.0\textbackslash{}1.0\textbackslash{}1.0\tabularnewline
Single-hop, Fixed-route &  & 11 (6\%) & NA & NA &  & 8 (4\%) & NA & 1.0\textbackslash{}1.4\textbackslash{}3.1\tabularnewline
Multi-hop, Fixed-route &  & 16 (8\%) & 0.0\textbackslash{}0.6\textbackslash{}1.9 & 0.0\textbackslash{}0.4\textbackslash{}2.5 &  & 18 (9\%) & NA & 1.0\textbackslash{}1.4\textbackslash{}3.2\tabularnewline
Single-hop, Flexible-route &  & 32 (16\%) & NA & NA &  & 26(13\%) & 0.0\textbackslash{}3.0\textbackslash{}8.1 & 1.0\textbackslash{}1.2\textbackslash{}2.3\tabularnewline
Multi-hop, Flexible-route &  & 52 (26\%) & 0.0\textbackslash{}0.5\textbackslash{}2.6 & 0.0\textbackslash{}0.6\textbackslash{}2.1 &  & 62(31\%) & 0.0\textbackslash{}2.1\textbackslash{}6.9 & 1.0\textbackslash{}1.3\textbackslash{}1.8\tabularnewline
\hline 
\end{tabular}
\end{table}

\begin{table}
\footnotesize
\caption{Quality of service measures for the five matching methods in a system
with relatively high spatiotemporal proximity among trips\label{tab:Quality-of-service-multihop-high_Spatiotemporal}}
\tabcolsep .22pc
\centering{}%
\begin{tabular}{llccccccc}
\hline 
 &  & \multicolumn{3}{c}{Riders} &  &  & \multicolumn{2}{c}{Drivers}\tabularnewline
\cline{3-5} \cline{7-9} 
 &  & Num. (\%) & Min\textbackslash{}avg.\textbackslash{}max  & Min\textbackslash{}avg.\textbackslash{}max  &  & Num. (\%) & Min\textbackslash{}avg.\textbackslash{}max  & Min\textbackslash{}avg.\textbackslash{}max \tabularnewline
Matching method &  & served & \ num. transfers & \ wait in transfer (min) &  & involved & \ extra travel time (min) & riders on board\tabularnewline
\hline 
\hline 
OD-based &  & 21(11\%) & NA & NA &  & 19(10\%) & NA & 1.0\textbackslash{}1.1\textbackslash{}2.4\tabularnewline
Single-hop, Fixed-route &  & 79(40\%) & NA & NA &  & 62(31\%) & NA & 1.0\textbackslash{}1.3\textbackslash{}4.2\tabularnewline
Multi-hop, Fixed-route &  & 89(45\%) & 0.0\textbackslash{}1.5\textbackslash{}2.5 & 0.0\textbackslash{}0.9\textbackslash{}3 &  & 104(52\%) & NA & 1.0\textbackslash{}1.9\textbackslash{}3.1\tabularnewline
Single-hop, Flexible-route &  & 104(52\%) & NA & NA &  & 79(35\%) & 0.0\textbackslash{}2.5\textbackslash{}6.1 & 1.0\textbackslash{}1.3\textbackslash{}2.9\tabularnewline
Multi-hop, Flexible-route &  & 152(76\%) & 0.0\textbackslash{}2.2\textbackslash{}2.9 & 0.0\textbackslash{}1.3\textbackslash{}5 &  & 181(86\%) & 0.0\textbackslash{}4.7\textbackslash{}8.8 & 1.0\textbackslash{}1.3\textbackslash{}3.6\tabularnewline
\hline 
\end{tabular}
\end{table}
\end{landscape}

The problem instances whose performance we studied in Table \ref{tab:Quality-of-service-multihop}
represent a ridesharing system that is spatiotemporally sparse. To
compare the matching methods in a more realistic setting, we generated
10 random instances of a ridesharing system with higher spatiotemporal
proximity among trips (Table \ref{tab:Quality-of-service-multihop-high_Spatiotemporal}). Similar to the problem instance studied in
Table \ref{tab:Quality-of-service-multihop}, this problem instance
contains 400 participants, 200 of which are riders and 200 are drivers.
The earliest departure times of all trips is generated within a 30
minute time period. The network is clustered into two sets of non-intersecting
sections, where all trip origins are randomly selected from stations
located in one of the sections, and all trip destinations from the
other. This trip generation procedure significantly increases the
spatiotemporal proximity among trips. The number of served riders
and matched drivers using all five ride-matching methods along with
some measures of quality of service are demonstrated in Table \ref{tab:Quality-of-service-multihop-high_Spatiotemporal}.
This table suggests that all matching methods serve higher number
of riders in a more spatiotemporally-dense ridesharing system.

\subsection{Percentage of Satisfied Riders' Requests}

Figure \ref{fig:Served_Riders} displays the contour plots of the
percentage of satisfied riders for the 420 randomly generated problem
instances. As intuition suggests, for a given number of riders, the
percentage of served requests increases with the number of drivers
in the system. Analyzing this type of graph can be especially helpful
when one is interested in the critical mass of participants required
to keep the system up and running, or when looking to set marketing
strategies. For example, if we have 100 registered riders in a network
where the origin and destination of trips are completely random, and
aim to serve at least 40\% of the unfiltered riders, then we know
we need at least 230 drivers (330 participants), and can tailor our
marketing strategies accordingly. 

Note that the percentages shown in Figure \ref{fig:Served_Riders}
are the result of uniformly distributed OD patterns in the problem instances,
which is the worst case scenario with a much less chance of joint rides
than found in real networks, where trip distributions imply high demands
for certain OD pairs, hence increasing the spatial proximity of trips
in the network. It is expected that real networks can operate with
much better ratios between drivers and riders. As a consequence, marketing
attempts in real networks can be targeted towards the more popular
OD pairs. 

\begin{figure}
\begin{centering}
\includegraphics[width=3.3in]{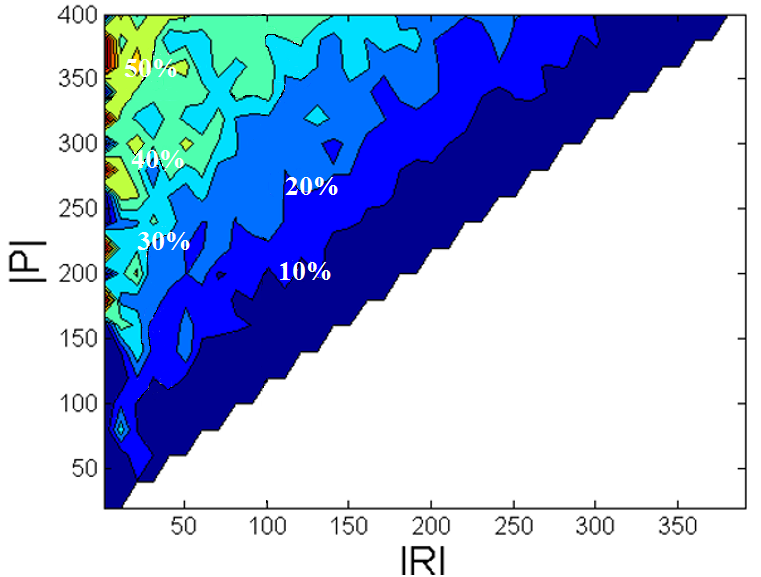}
\par\end{centering}
\caption{Percentage of un-filtered riders served\label{fig:Served_Riders}}
\end{figure}

\subsection{Algorithm Performance\label{sub:Algorithm-Performance}}

Figures \ref{fig:time_MIP} and \ref{fig:time_decompo} compare the
time required to solve the randomly generated ride-matching problems
to optimality, without and with the decomposition algorithm, respectively. Solution
times in these figures were obtained by solving models (\ref{equ:determinstic})
and (\ref{equ:determinstic-1}), respectively, using the CPLEX optimization
engine. Comparison of these two figures suggests considerable savings
in solution times can be obtained using the decomposition algorithm,
without any trade-offs in terms of solution accuracy. 
Note that solution times reported in Figures \ref{fig:time_MIP} and \ref{fig:time_decompo} have been obtained for solving the optimization problems after applying the pre-processing procedure. The majority of the raw versions of the 420 problem instances cannot be solved without reducing the size of the problem instances (will run out of memory.) 

Figure \ref{fig:Iteration} shows the number of iterations of the
decomposition algorithm it takes to solve the problems to optimality.
Note that the higher numbers of iterations correspond to problems
with higher solution times. None of the problem instances require
more than 7 iterations to be solved to optimality using the decomposition
algorithm.

\begin{figure}
\begin{centering}
\subfloat[Original matching problem solution times (sec)\label{fig:time_MIP}]{\begin{centering}
\includegraphics[width=3in]{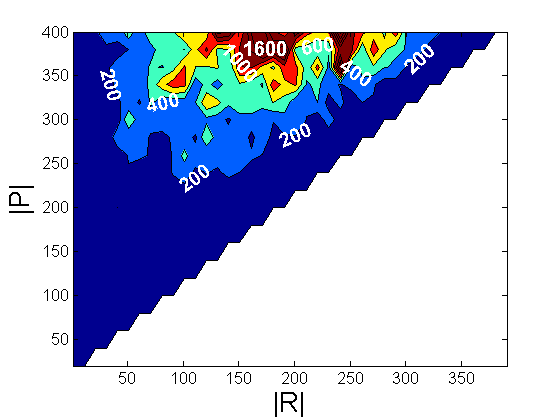}
\par\end{centering}
}\hfill{}\subfloat[Decomposition algorithm solution times (sec)\label{fig:time_decompo}]{\begin{centering}
\includegraphics[width=3in]{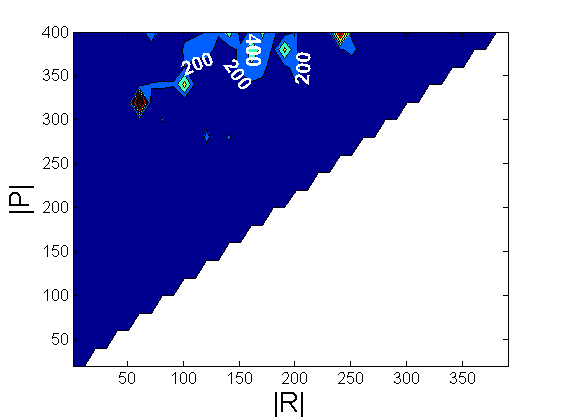}
\par\end{centering}
}
\par\end{centering}
\caption{Improvements in solution times. Contour plots of solution times in
seconds.\label{fig:alg_performance}}
\end{figure}
\begin{figure}
\begin{centering}
\includegraphics[width=3.3in]{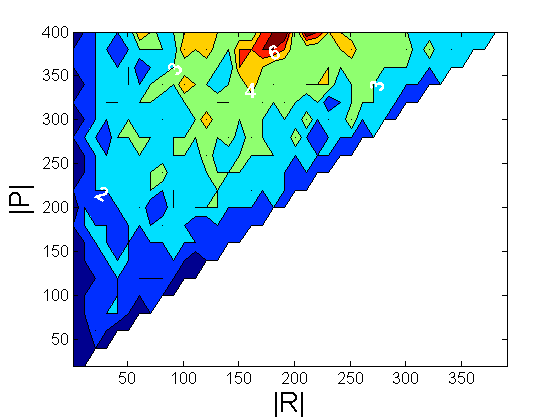}
\par\end{centering}
\caption{Number of decomposition algorithm iterations\label{fig:Iteration}}
\end{figure}

\subsection{Transfers}

Transfers usually connote discomfort, and understandably so; since
public transportation commonly runs under a fixed schedule, transit users might have to experience longer travel times and endure high waiting times in transfers. Note, however, that ridesharing systems of the future could be fundamentally different in nature, and better user-acceptance of transfers is possible depending on the ease of transfers, coordination of transfers in time, and reduction in transfer times offered by such systems \citep{taylor2009thinking}. This will naturally depend on the market penetration of the system as well as user-side behavioral changes, which are out of the scope of this paper. While this remains speculative, in this study we allow riders to specify their maximum number of transfers, and incorporate the maximum ride times requested by riders and drivers in our models.
These input parameters could vary among individuals depending on their
accessibility to personal vehicles, and their values of time. As an individual
riders' own transfer preferences are incorporated, this will help ease concerns on the perceived quality of service being affected by transfers.

A further reason to model a multi-hop ridesharing system is to a provide a general framework that can combine ridesharing systems with other modes of transportation, as discussed in section \ref{sec:Literature-Review}. It is straight-forward to include, say, a transit system simply as a set of high-capacity ``drivers'' with fixed routes in the formulation. In such cases, more than one transfer between the transit system and the rideshare mode is reasonable (one transfer from the rideshare vehicle to transit, and a second one from the final transit station to the rideshare vehicle.) Such multi-modal applications highlight the need for
algorithms that do not systemically limit the number of transfers.

Figure \ref{fig:tr} shows the cumulative distribution of number of transfers
in the 420 problem instances solved in this section. Even though
in our numerical experiments we assume that riders are comfortable
with up to 3 transfers, this figure suggests that a considerable portion of trips can be served with zero to one transfer, and riders in over 90\% of the problem instances can be fully served with no more than two transfers, which has serious implications when it comes to implementing a ridesharing system. This figure also suggests that transfers become necessary when the number of drivers is relatively low compared to the number of riders, as intuition suggests.

\begin{figure}
\begin{centering}
\subfloat[Percentage of served riders with zero transfers\label{fig:tr_zero}]{\begin{centering}
\includegraphics[width=3in]{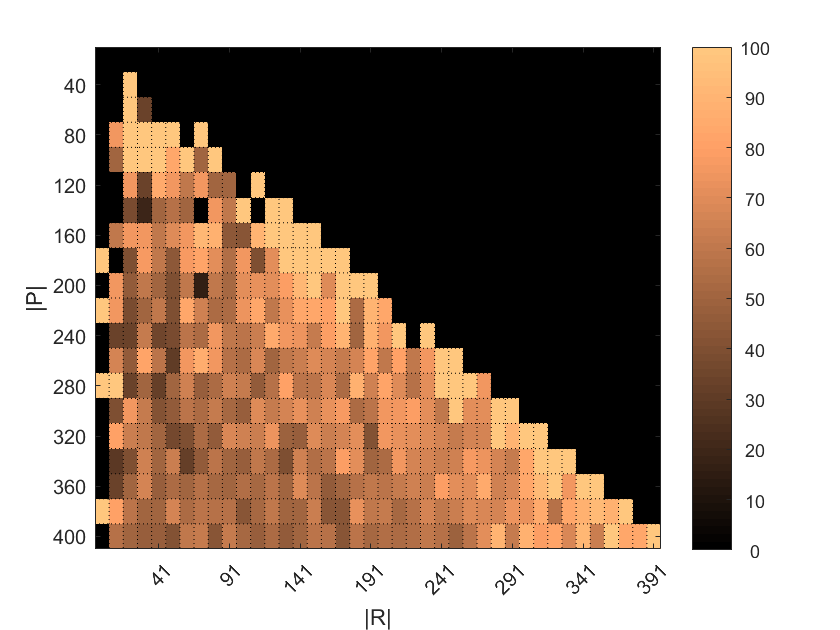}
\par\end{centering}
}\hfill{}\subfloat[Percentage of served riders with one or fewer transfers\label{fig:tr_one}]{\begin{centering}
\includegraphics[width=3in]{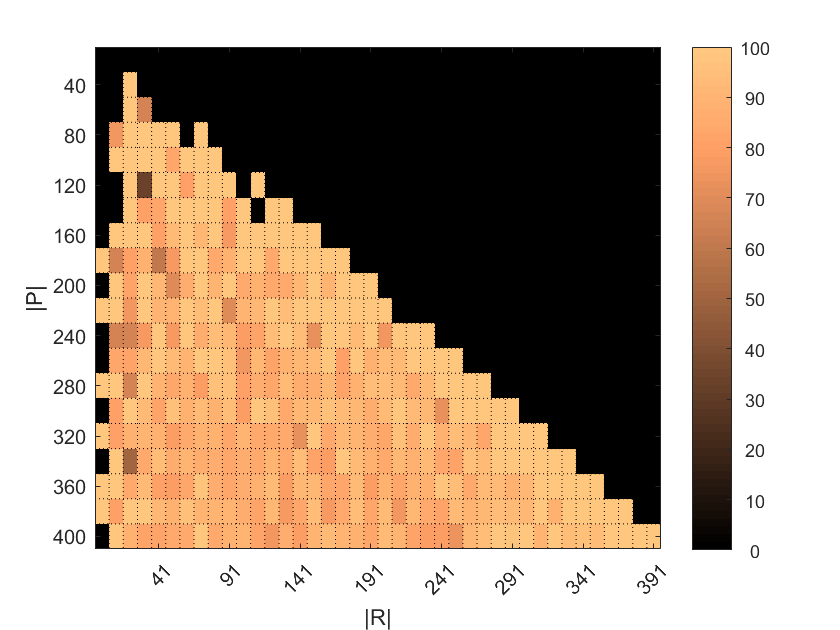}
\par\end{centering}
}
\par\end{centering}
\begin{centering}
\subfloat[Percentage of served riders with two or fewer transfers\label{fig:tr_two}]{\begin{centering}
\includegraphics[width=3in]{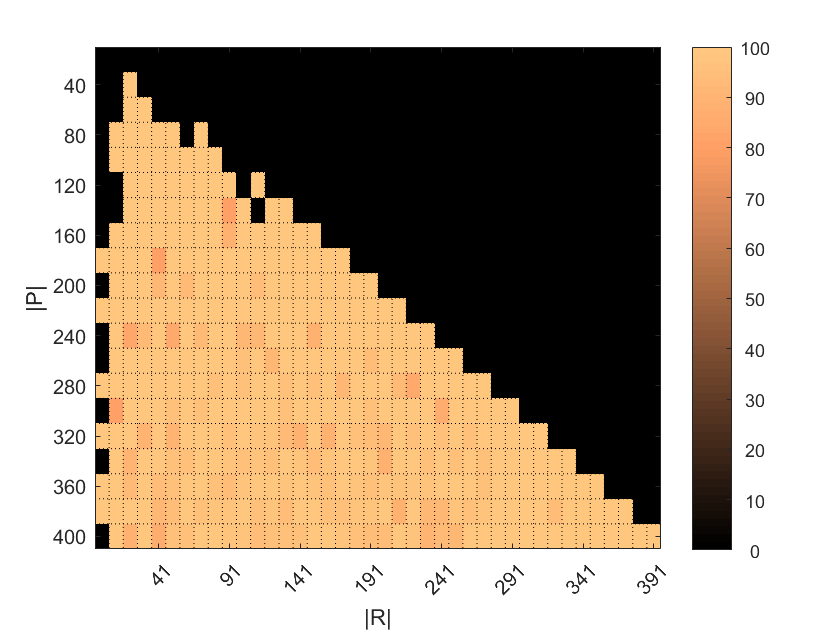}
\par\end{centering}
}\hfill{}\subfloat[Percentage of served riders with three or fewer transfers\label{fig:tr_three}]{\begin{centering}
\includegraphics[width=3in]{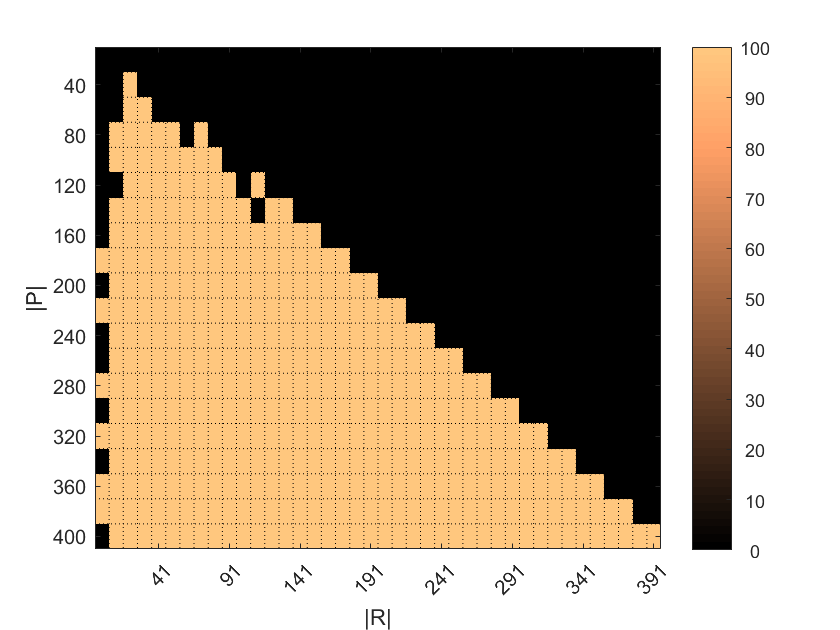}
\par\end{centering}
}
\par\end{centering}
\caption{Distribution of transfers\label{fig:tr}}
\end{figure}

\subsection{Sensitivity Analysis}\label{sec:sensitivity}

In order to perform sensitivity analysis over some of the parameters
of the ride-matching problem, we generate and solve two additional
series of matching problems. All these problems contain 400 participants,
with equal number of riders and drivers. This configuration is selected
because Figures \ref{fig:Served_Riders} and \ref{fig:alg_performance}
suggest that problems with equal number of riders and drivers produce
the highest number of matches and require the most computational
resources. For the following problems, we maximize an objective function
that includes two main terms: sum of the total number of served riders,
($\sum_{r\in R}z_{r}$), and sum of the weighted travel times by riders, $(-\sum_{r\in R}W_{r}\sum_{d\in D:(r,d)\in M}\sum_{\ell=(t_{i},s_{i},t_{j},s_{j})\in L_{rd}}(t_{j}-t_{i})y_{\ell}^{rd})$.
The term $W_{r}=1/(T_{r}^{TB}+1)$ (where $T_{r}^{TB}$ is the maximum
ride time for rider $r$) is considered as the weight for the travel
time of rider $r$ in this objective function to ensure that maximizing
the total number of served riders remains the main priority of the
system. In the extreme scenario under which all riders use their entire travel time budget, this term ensures that the expression $\sum_{r\in R}W_{r}\sum_{d\in D:(r,d)\in M}\sum_{\ell=(t_{i},s_{i},t_{j},s_{j})\in L_{rd}}(t_{j}-t_{i})y_{\ell}^{rd}$ remains less than $1$, prioritizing the objective of maximizing the total number of served riders.

Table \ref{tab:Sensitivity_49} shows the performance of the decomposition
algorithm in solving 5 problem instances in a network with 49 stations. 
The field ``release period'' in this table indicates the period of time
during which the earliest departure times of participants are generated.
For this field, we used two values of 30 and 60 minutes. Problem instances
with release period of 30 minutes incorporate higher temporal proximity
between trips. 

The field ``Dir. trips'' in Table \ref{tab:Sensitivity_49} has
a binary value. Value of zero for this field is used when the trip
origin and destination locations are generated randomly in the network.
In problem instances with value of 1 for this field, we use a clustered
network. In a clustered network, trips experience a higher degree
of spatial proximity. 

Finally, we used two values of 1.1 and 1.2 as participants'
travel time budget factors to represent their flexibility on how much time they are willing to spend in the network. Value of 1.1/1.2 for this parameter indicates that
all participants are assumed to have maximum ride times that are 1.1/1.2
times their shortest path travel times. 

The field ``Optimal'' in Table \ref{tab:Sensitivity_49} indicates
whether a problem was solved to optimality (Y), or we ran out
of memory when solving at least one of the sub-problems (the original
problem in case of a static problem), and had to report the solution
obtained by the heuristic described in section \ref{sub:Heuristic-solutions}.

All these problem instances have been solved as both static problems,
and dynamic problems with different re-optimization periods. For each problem instance, ``NA'' under the ``Re-optimization period''
indicates that the problem is solved in a static setting, assuming
all participants to have registered their trips before the onset of the planning horizon when the matching problem is solved. Other values under this field show the re-optimization period. We assume that the entire demand for the re-optimization period $i$ arrives within the window $[(i-1)k,ik]$, where $k$ is the length of the re-optimization period. In addition, after each re-optimization period, we fix the itineraries of all matched riders and drivers, and include the fixed itineraries of matched drivers in the problems
solved for the following re-optimization periods, if their travel
time windows allow. 

Not surprisingly, static versions of all five problem instances have
the highest matching rates and the largest solution times compared
to their dynamic counterparts, since static problems include all participants.
Following the same logic, given a fixed number of participants, longer
re-optimization periods translate into a larger pool of participants
in the matching problems, and higher matching rates. 

For each problem, we provide statistics on the number of rider transfers
(based on the itineraries of matched riders), and the waiting time
of riders during transfers (based on the itineraries of riders who
experience transfers.) As a general trend, all these values decrease
with the length of the re-optimization period, although the average
waiting time of riders during transfers is in general negligible,
especially given the potential uncertainties in travel times. 

Table \ref{tab:Sensitivity_49} also reports the number of matched
drivers, and the statistics on the extra travel time by the matched drivers
(compared to their shortest path travel times.) This table suggests
that the number of matched drivers decreases with the length of the re-optimization period. However, no monotonic relationship can be observed between the drivers' extra time in the network and the length of the re-optimization period. 

Compared to the first problem instance, problem instances 4, 5, and
6 incorporate higher temporal, spatial, and spatiotemporal proximity
among trips, respectively. Although all these problem instances lead
to higher number of served riders compared to the first problem instance,
it is interesting to note that the effect of higher spatial proximity
on the number of served riders is more significant than that of higher
temporal proximity. An interesting point to notice is that at 5-min
re-optimization periods, all the problem instances can be solved in
less than 1 minute. 

The highest number of served riders is obtained in problem instance
2, where the participants' time budgets is the highest. However, notice
that the measures of quality of service, namely the average extra
travel time for drivers, and the average number of transfers and transfer
wait time for riders, are also considerably higher compared to other
problem instances. The matching rate and quality of service measures
reported on problem instance 3 in which riders are (naturally) assumed
to have a higher maximum ride time than drivers lie between those
of instances 1 and 2.

Table \ref{tab:Sensitivity_100} displays problem instances similar
to the ones in Table \ref{tab:Sensitivity_49} in terms of problem
parameters, but in a network with 100 stations. We should point out
that we could not solve the static version of problem instance 5 to
optimality (ran out of memory), and instead used the heuristic approach
described in section \ref{sub:Heuristic-solutions}. The two numbers
reported under the ``Num. served'' field show the lower and upper
bounds on the optimal solution (the lower bound solution is the feasible
solution that can be eventually used). Note that if such a problem
is encountered by a system due to lack of computational capacity,
they could opt for re-optimizing the system more periodically. 

In general, Table \ref{tab:Sensitivity_100} demonstrates the same
trends as in Table \ref{tab:Sensitivity_49}. The most prominent difference
that can be observed comparing the two tables is that Table \ref{tab:Sensitivity_100}
has fewer number of served riders. This is an expected result, since
by increasing the number of stations and keeping the same number of
participants, we are in fact decreasing the spatial proximity among
trips. In practice, however, increased number of stations could lead
to potentially higher levels of demand, as higher number of stations
translates into higher accessibility to the ridesharing system.

\begin{landscape}
\begin{table}
\footnotesize
\tabcolsep .22pc
\caption{Sensitivity study over the problem parameters. All instances are generated
with 400 participants composing of 200 riders and 200 drivers in a
randomly generated network with 49 stations\label{tab:Sensitivity_49}}
\centering{}%
\begin{tabular}{cccccccccccccc}
\hline 
 &  &  &  &  & \multicolumn{3}{c}{Riders} &  & \multicolumn{2}{c}{Drivers} &  & \multicolumn{2}{c}{Solution}\tabularnewline
\cline{6-8} \cline{10-11} \cline{13-14} 
Problem & Dir. & Release & Time budget & Re-optimization & Num.  & Min\textbackslash{}avg.\textbackslash{}max  & Min\textbackslash{}avg.\textbackslash{}max  &  & Num. & Min\textbackslash{}avg.\textbackslash{}max  &  & time & Optimal\tabularnewline
Instance & trips & period & rider/driver & period  & served & num. transfers & wait in transfer  &  & Involved  & extra time  &  & (sec) & \textit{Y/LB}\tabularnewline
 &  &  &  & (min) &  &  & (min) &  &  & (min) &  &  & \tabularnewline
\hline 
\hline 
$1$ & $0$ & $60$ & $1.1/1.1$ & NA & 69 & 0.0\textbackslash{}0.9\textbackslash{}3.0 & 0.0\textbackslash{}0.6\textbackslash{}3.0 &  & 98 & 1.0\textbackslash{}2.7\textbackslash{}3.0 &  & 518 & \textit{Y}\tabularnewline
$1$ & $0$ & $60$ & $1.1/1.1$ & $10$ & 27 & 0.2\textbackslash{}0.5\textbackslash{}1.0 & 0.0\textbackslash{}0.3\textbackslash{}1.8 &  & 44 & 1.8\textbackslash{}3.0\textbackslash{}3.4 &  & 16 & \textit{Y}\tabularnewline
$1$ & $0$ & $60$ & $1.1/1.1$ & $5$ & 20 & 0.4\textbackslash{}0.4\textbackslash{}0.8 & 0.0\textbackslash{}0.4\textbackslash{}1.2 &  & 36 & 1.1\textbackslash{}2.2\textbackslash{}3.2 &  & 6 & \textit{Y}\tabularnewline
\hline 
$2$ & $0$ & $60$ & $1.2/1.2$ & NA & 150 & 0.0\textbackslash{}2.1\textbackslash{}2.1 & 0.0\textbackslash{}2.5\textbackslash{}4.1 &  & 144 & 0.0\textbackslash{}11.2\textbackslash{}15 &  & 1920 & \textit{Y}\tabularnewline
$2$ & $0$ & $60$ & $1.2/1.2$ & $10$ & 72 & 0.0\textbackslash{}0.6\textbackslash{}1.7 & 0.0\textbackslash{}1.3\textbackslash{}3.1 &  & 105 & 2.9\textbackslash{}8.5\textbackslash{}11 &  & 174 & \textit{Y}\tabularnewline
$2$ & $0$ & $60$ & $1.2/1.2$ & $5$ & 67 & 0.0\textbackslash{}0.4\textbackslash{}1.7 & 0.0\textbackslash{}0.4\textbackslash{}3.3 &  & 87 & 3.6\textbackslash{}6.5\textbackslash{}9.2 &  & 29 & \textit{Y}\tabularnewline
\hline 
$3$ & $0$ & $60$ & $1.2/1.1$ & NA & 112 & 0.0\textbackslash{}0.7\textbackslash{}2.4 & 0.0\textbackslash{}1.5\textbackslash{}3.1 &  & 123 & 0.1\textbackslash{}3.1\textbackslash{}5.4 &  & 718 & \textit{Y}\tabularnewline
$3$ & $0$ & $60$ & $1.2/1.1$ & $10$ & 43 & 0.0\textbackslash{}0.4\textbackslash{}1.2 & 0.0\textbackslash{}0.8\textbackslash{}3.0 &  & 80 & 2.0\textbackslash{}3.0\textbackslash{}5.1 &  & 53 & \textit{Y}\tabularnewline
$3$ & $0$ & $60$ & $1.2/1.1$ & $5$ & 38 & 0.0\textbackslash{}0.4\textbackslash{}0.9 & 0.0\textbackslash{}0.4\textbackslash{}2.1 &  & 67 & 2.6\textbackslash{}4.2\textbackslash{}4.2 &  & 11 & \textit{Y}\tabularnewline
\hline 
$4$ & $0$ & $30$ & $1.1/1.1$ & NA & 87 & 0.0\textbackslash{}0.9\textbackslash{}30 & 0.0\textbackslash{}0.7\textbackslash{}1.5 &  & 115 & 0.75\textbackslash{}2.3\textbackslash{}3.3 &  & 2357 & \textit{Y}\tabularnewline
$4$ & $0$ & $30$ & $1.1/1.1$ & $10$ & 54 & 0.0\textbackslash{}0.7\textbackslash{}2.3 & 0.0\textbackslash{}0.6\textbackslash{}3.3 &  & 80 & 1.7\textbackslash{}3.1\textbackslash{}4.2 &  & 201 & \textit{Y}\tabularnewline
$4$ & $0$ & $30$ & $1.1/1.1$ & $5$ & 48 & 0.0\textbackslash{}0.4\textbackslash{}1.5 & 0.0\textbackslash{}0.5\textbackslash{}3.7 &  & 68 & 0.5\textbackslash{}2.5\textbackslash{}4.3 &  & 12 & \textit{Y}\tabularnewline
\hline 
$5$ & $1$ & $60$ & $1.1/1.1$ & NA & 137 & 0.0\textbackslash{}2.1\textbackslash{}3.0 & 0.0\textbackslash{}1.7\textbackslash{}7.0 &  & 158 & 0.0\textbackslash{}2.4\textbackslash{}8.0 &  & 1104 & \textit{Y}\tabularnewline
$5$ & $1$ & $60$ & $1.1/1.1$ & $10$ & 73 & 0.0\textbackslash{}0.8\textbackslash{}2.3 & 0.0\textbackslash{}0.6\textbackslash{}2.8 &  & 107 & 0.0\textbackslash{}2.4\textbackslash{}7.06 &  & 52 & \textit{Y}\tabularnewline
$5$ & $1$ & $60$ & $1.1/1.1$ & $5$ & 57 & 0.0\textbackslash{}0.6\textbackslash{}1.7 & 0.0\textbackslash{}0.6\textbackslash{}1.8 &  & 85 & 0.3\textbackslash{}3.0\textbackslash{}3.4 &  & 22 & \textit{Y}\tabularnewline
\hline 
$6$ & $1$ & $30$ & $1.1/1.1$ & NA & 152 & 0.0\textbackslash{}2.2\textbackslash{}3.0 & 0.0\textbackslash{}1.3\textbackslash{}8.0 &  & 192 & 0.0\textbackslash{}2.3\textbackslash{}8.0 &  & 2750 & \textit{Y}\tabularnewline
$6$ & $1$ & $30$ & $1.1/1.1$ & $10$ & 116 & 0.0\textbackslash{}0.8\textbackslash{}2.6 & 0.0\textbackslash{}0.5\textbackslash{}3.7 &  & 133 & 0.0\textbackslash{}2.5\textbackslash{}8.0 &  & 142 & \textit{Y}\tabularnewline
$6$ & $1$ & $30$ & $1.1/1.1$ & $5$ & 101 & 0.0\textbackslash{}0.6\textbackslash{}1.5 & 0.0\textbackslash{}0.4\textbackslash{}2.2 &  & 112 & 0.1\textbackslash{}1.9\textbackslash{}6.2 &  & 48 & \textit{Y}\tabularnewline
\hline 
\end{tabular}
\end{table}
\end{landscape}

\begin{landscape}
\begin{table}
\footnotesize
\tabcolsep .22pc
\caption{Sensitivity study over the problem parameters. All instances are generated
with 400 participants composing of 200 riders and 200 drivers in a
randomly generated network with 100 stations\label{tab:Sensitivity_100}}
\centering{}%
\begin{tabular}{cccccccccccccc}
\hline 
 &  &  &  &  & \multicolumn{3}{c}{Riders} &  & \multicolumn{2}{c}{Drivers} &  & \multicolumn{2}{c}{Solution}\tabularnewline
\cline{6-8} \cline{10-11} \cline{13-14} 
Problem & Dir. & Release & Time budget & Re-optimization & Num.  & Min\textbackslash{}avg.\textbackslash{}max  & Min\textbackslash{}avg.\textbackslash{}max  &  & Num. & Min\textbackslash{}avg.\textbackslash{}max  &  & time & Optimal\tabularnewline
Instance & trips & period & rider/driver & period  & served & num. transfers & wait in transfer  &  & Involved  & extra time  &  & (sec) & \textit{Y/LB}\tabularnewline
 &  &  &  & (min) &  &  & (min) &  &  & (min) &  &  & \tabularnewline
\hline 
\hline 
$1$ & $0$ & $60$ & $1.1/1.1$ & NA & 61 & 0.0\textbackslash{}1.8\textbackslash{}2.7 & 0.0\textbackslash{}1.8\textbackslash{}10 &  & 94 & 0.0\textbackslash{}4.4\textbackslash{}7.0 &  & 501 & \textit{Y}\tabularnewline
$1$ & $0$ & $60$ & $1.1/1.1$ & $10$ & 18 & 0.5\textbackslash{}1.1\textbackslash{}1.8 & 0.2\textbackslash{}1.1\textbackslash{}3.7 &  & 39 & 0.5\textbackslash{}5.9\textbackslash{}9.1 &  & 40 & \textit{Y}\tabularnewline
$1$ & $0$ & $60$ & $1.1/1.1$ & $5$ & 11 & 0.3\textbackslash{}0.5\textbackslash{}0.7 & 0.1\textbackslash{}0.4\textbackslash{}0.8 &  & 20 & 0.3\textbackslash{}3.9\textbackslash{}8.2 &  & 36 & \textit{Y}\tabularnewline
\hline 
$2$ & $0$ & $60$ & $1.2/1.2$ & NA & 138 & 0.0\textbackslash{}2.4\textbackslash{}2.5 & 0.3\textbackslash{}2.4\textbackslash{}6.2 &  & 124 & 0.0\textbackslash{}16.2\textbackslash{}20.2 &  & 1339 & \textit{Y}\tabularnewline
$2$ & $0$ & $60$ & $1.2/1.2$ & $10$ & 59 & 0.1\textbackslash{}1.6\textbackslash{}2.1 & 0.0\textbackslash{}1.0\textbackslash{}5.1 &  & 83 & 1.0\textbackslash{}14.3\textbackslash{}19.8 &  & 515 & \textit{Y}\tabularnewline
$2$ & $0$ & $60$ & $1.2/1.2$ & $5$ & 39 & 0.0\textbackslash{}0.4\textbackslash{}0.9 & 0.1\textbackslash{}0.9\textbackslash{}3.9 &  & 63 & 2.1\textbackslash{}12.2\textbackslash{}18.2 &  & 316 & \textit{Y}\tabularnewline
$2$ & $0$ & $60$ & $1.2/1.2$ & $2$ & 32 & 0.0\textbackslash{}0.3\textbackslash{}0.6 & 0.0\textbackslash{}0.7\textbackslash{}3.1 &  & 40 & 0.1\textbackslash{}8.1\textbackslash{}14.0 &  & 115 & \textit{Y}\tabularnewline
\hline 
$3$ & $0$ & $60$ & $1.1/1.2$ & NA & 85 & 0.1\textbackslash{}1.8\textbackslash{}2.5 & 0.3\textbackslash{}2.1\textbackslash{}4.1 &  & 102 & 0.0\textbackslash{}2.1\textbackslash{}6.2 &  & 1020 & \textit{Y}\tabularnewline
$3$ & $0$ & $60$ & $1.1/1.2$ & $10$ & 35 & 0.3\textbackslash{}1.5\textbackslash{}1.9 & 0.1\textbackslash{}1.2\textbackslash{}4.0 &  & 49 & 0.2\textbackslash{}3.9\textbackslash{}7.1 &  & 180 & \textit{Y}\tabularnewline
$3$ & $0$ & $60$ & $1.1/1.2$ & $5$ & 21 & 0.2\textbackslash{}0.4\textbackslash{}1.1 & 0.1\textbackslash{}0.8\textbackslash{}3.8 &  & 30 & 0.1\textbackslash{}3.9\textbackslash{}7.3 &  & 115 & \textit{Y}\tabularnewline
\hline 
$4$ & $0$ & $30$ & $1.1/1.1$ & $2$ & 36 & 0.2\textbackslash{}1.0\textbackslash{}1.5 & 0.2\textbackslash{}0.8\textbackslash{}2.0 &  & 63 & 0.1\textbackslash{}4.5\textbackslash{}8.2 &  & 866 & \textit{Y}\tabularnewline
$4$ & $0$ & $30$ & $1.1/1.1$ & $10$ & 23 & 0.1\textbackslash{}0.8\textbackslash{}1.7 & 0.3\textbackslash{}1.7\textbackslash{}3.4 &  & 42 & 0.1\textbackslash{}4.1\textbackslash{}9.1 &  & 339 & \textit{Y}\tabularnewline
$4$ & $0$ & $30$ & $1.1/1.1$ & $5$ & 15 & 0.0\textbackslash{}0.4\textbackslash{}0.7 & 0.3\textbackslash{}0.7\textbackslash{}2.3 &  & 23 & 0.1\textbackslash{}3.2\textbackslash{}7.6 &  & 31 & \textit{Y}\tabularnewline
\hline 
$5$ & $1$ & $60$ & $1.1/1.1$ & NA & 132 & 0.5\textbackslash{}2.4\textbackslash{}2.9 & 0.4\textbackslash{}1.9\textbackslash{}4.1 &  & 145 & 0.0\textbackslash{}6.7\textbackslash{}8.9 &  & 856 & \textit{Y}\tabularnewline
$5$ & $1$ & $60$ & $1.1/1.1$ & $10$ & 29 & 0.1\textbackslash{}0.8\textbackslash{}1.2 & 0.0\textbackslash{}0.7\textbackslash{}2.1 &  & 49 & 0.1\textbackslash{}4.0\textbackslash{}7.1 &  & 124 & \textit{Y}\tabularnewline
$5$ & $1$ & $60$ & $1.1/1.1$ & $5$ & 19 & 0.1\textbackslash{}0.7\textbackslash{}0.9 & 0.1\textbackslash{}0.4\textbackslash{}1.9 &  & 32 & 0.2\textbackslash{}4.5\textbackslash{}7.1 &  & 30 & \textit{Y}\tabularnewline
\hline 
$6$ & $1$ & $30$ & $1.1/1.1$ & NA & [78 128] & 0.4\textbackslash{}2.4\textbackslash{}2.7 & 0.1\textbackslash{}1.3\textbackslash{}3.8 &  & 161 & 0.2\textbackslash{}6.2\textbackslash{}10.1 &  & 753 & \textit{LB}\tabularnewline
$6$ & $1$ & $30$ & $1.1/1.1$ & $10$ & 61 & 0.1\textbackslash{}1.0\textbackslash{}1.2 & 0.1\textbackslash{}1.0\textbackslash{}2.1 &  & 99 & 0.3\textbackslash{}4.9\textbackslash{}6.9 &  & 2140 & \textit{Y}\tabularnewline
$6$ & $1$ & $30$ & $1.1/1.1$ & $5$ & 40 & 0.1\textbackslash{}0.8\textbackslash{}1.2 & 0.1\textbackslash{}0.8\textbackslash{}1.9 &  & 64 & 0.1\textbackslash{}2.7\textbackslash{}4.5 &  & 447 & \textit{Y}\tabularnewline
$6$ & $1$ & $30$ & $1.1/1.1$ & $2$ & 28 & 0.0\textbackslash{}0.5\textbackslash{}0.7 & 0.1\textbackslash{}0.7\textbackslash{}1.5 &  & 42 & 0.0\textbackslash{}1.0\textbackslash{}3.1 &  & 123 & \textit{Y}\tabularnewline
\hline 
\end{tabular}
\end{table}
\end{landscape}


\subsubsection{Maximum Ride Time}

Since drivers in a risharing system are traveling based on their own schedules, it is safe to assume that they are less likely to have significant travel time budget factors (i.e., the total time one is willing to spend in the network as a percentage of one's shortest path travel time). Riders in a ridesharing system, on the other hand, are more likely to consider rather larger travel time budget factors due to their possible lack of access to a private vehicle and/or the low frequency and potentially longer travel times of public transportation.
It is reasonable to assume that increasing the travel time budget factor of participants would improve the matching rate by allowing larger detours from one's shortest path and more flexibility in one's  departure and arrival times, but it would at the same time increase solution times as a result of larger link sets in the optimization problem. Figure \ref{fig:budget} demonstrates this trade-off.  

All the problem instances in Figure \ref{fig:budget} include 200 riders and 200 drivers. The travel time budget factor of 1.1 is considered for drivers, and the rest of the parameter values are the default values from section \ref{sec:numericalstudy}.
Figure \ref{fig:budget} suggests that as we increase the rider travel time budget factor, both the matching rate and solution time demonstrate an increasing trend. While solution times increase exponentially with riders' travel time budget factor, the increase in the matching rate seems to have a higher rate initially, up to the travel budget factor of 1.3, after which the rate of increase in the matching rate decreases. This is due to the fact that drivers' availabilty in the network is bounded by their rather small travel time budget factor of 1.1, and although increasing the riders' travel time budget factor helps improve the matching rate to some degree, higher improvement in the matching rate can be obtained by increasing driver-hours spent in the network (e.g., if current drivers stay in the network for longer and/or more drivers join the system.)

\begin{figure}
\begin{centering}
\includegraphics[width=3in]{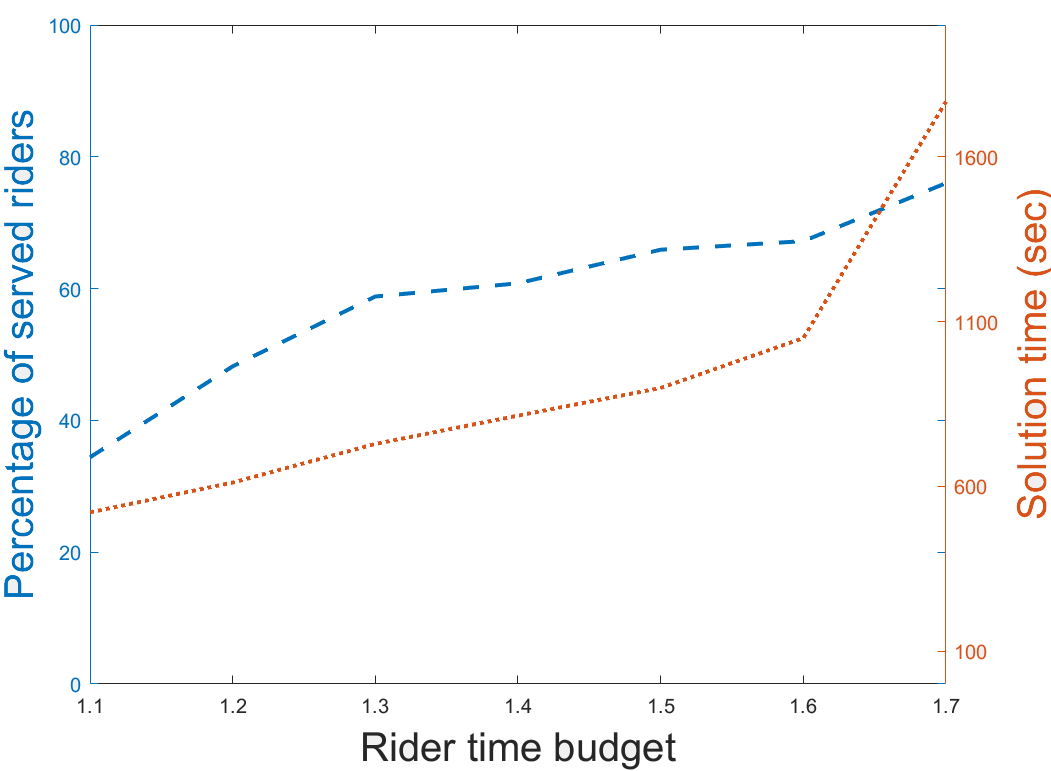}
\par\end{centering}
\caption{Impact of rider and driver maximum ride times on the matching rate and solution time\label{fig:budget}}
\end{figure}

In the following section, we present methods to make the solution
algorithm more appropriate for dynamic implementations by reducing
the solution times even further. \FloatBarrier

\section{Application in Practice}

The scalability of the solution algorithm to solve the ride-matching
problem may become an issue in practice. Larger area of coverage by
ridesharing companies leads to larger number of stations, and potentially
larger number of participants. Larger problems can lead to higher
solution times and reduce the effectiveness of the proposed algorithm
in dynamic applications. 

Forming reduced graphs for participants can address this problem to
a large extent. Even though a ridesharing company could cover a very
large area, the distance traveled by participants is usually within
a limited (but not necessarily small) circumference of their homes,
i.e., within city limits, and therefore their reduced graphs contain
a limited number of stations. It is only in rarer cases that people
might require to, or be willing to, travel longer distances in a dynamic
rideshare system. Thus it may be expected that larger reduced graphs
would be rarer.

In order for the ridesharing system to cover longer (such as inter-city)
trips, the system can locate one or two important stations in each
city, and use only those set of stations for riders with long-distance
trips. In addition, drivers who are selected to be included in the
ride-matching problems with such riders can be limited to those who
are making inter-city trips themselves. Using drivers with shorter
trip lengths in such settings would not be practical due to the large
number of transfers that forming an itinerary with short-distance
drivers would require. 

In the interest of lower solution times, in this section we suggest
two different approaches to simplify the problem and obtain high
quality heuristic solutions to the ride-matching problem within a
reasonable time period.

\subsection{Restricting the Number of Transfer Stations}

In the ride-matching problem in model (\ref{equ:determinstic}), itineraries
are devised assuming transfers can occur at any station. In practice,
we can identify stations where transfers are more likely to happen,
and limit the set of transfer stations. Let us denote by $S_{T}$
the set of stations where transfers are allowed. The flow conservation
constraint (\ref{eq:r_balance}) should then be replaced by two sets
of constraints, (\ref{eq:Transfer1}) and (\ref{eq:Transfer2}). Constraint
set (\ref{eq:Transfer1}) and (\ref{eq:Transfer2}) both model flow
conservation. Constraint set (\ref{eq:Transfer1}) covers stations
where transfers are allowed, and (\ref{eq:Transfer2}) covers the
rest of the stations. 

\begin{equation}
\!\!\!\sum_{^{d\in D^{\prime}}}\!\!\!\sum_{^{_{\ell=(t_{i},s_{i},t,s)\in L}^{\qquad t_{i},s_{i}:}}}\!\!\!\!\!y_{\ell}^{rd}=\sum_{^{\quad d\in D^{\prime}}}\!\!\!\sum_{^{_{\ell=(t,s,t_{j},s_{j})\in L}^{\qquad t_{j},s_{j}:}}}\!\!\!\!\!y_{\ell}^{rd}\quad\quad\begin{array}{l}
\forall r\in R\\
\forall t\in T_{r}\\
\forall s\in S_{T}\backslash\{OS_{r}\cup DS_{r}\}
\end{array}\label{eq:Transfer1}
\end{equation}

\begin{equation}
\!\!\!\sum_{^{\:\:\:_{\ell=(t_{i},s_{i},t,s)\in L}^{\:\:\:\:\:\:\:\:\:t_{i},s_{i}:}}}\!\!\!\!\!y_{\ell}^{rd}=\!\!\!\sum_{^{_{\:\:\:\ell=(t,s,t_{j},s_{j})\in L}^{\:\:\:\:\:\:\:\:\:\:\:\:\:t_{j},s_{j}:}}}\!\!\!\!\!y_{\ell}^{rd}\quad\quad\quad\quad\qquad\qquad\begin{array}{l}
\forall r\in R,\forall d\in D\\
\forall t\in T_{r}\\
\forall s\in S\backslash\{S_{T}\cup OS_{r}\cup DS_{r}\}
\end{array}\label{eq:Transfer2}
\end{equation}

To study the impact of limiting transfers to certain stations on the
quality of the solutions and solution times, we experimented with
the 420 randomly generated problem instances. For each problem instance,
we studied the solution to model (\ref{equ:determinstic}) where all
stations are transfer stations, and identified the top 80\% of stations
where transfers occur. We then limited transfers to those stations
only. Figure \ref{fig:Stations} shows the ratio of reduction in solution
times (new solution times divided by the old solution times). The
computational time reduction can range from 0 to 60\%. It is interesting,
however, to observe that the higher savings in solution times are
obtained for problem instances with higher original solution times
(see Figure \ref{fig:time_MIP}). 

In terms of solution quality, number of served riders using the top
80\% transfer stations remained on average within 95\% of the number
of served riders in the original solution, suggesting that the trade-off
between the solution time and quality in dynamic settings may be in
favor of restricting the number of transfer stations. These results
are similar to the results found by \citeauthor{masson2014optimization}~(2014b).
In their study of a multi-modal system designed to carry goods using
a combination of excess bus capacities and city freighters, they found
that 60\% of the transfers took place in two bus stops that were closer
to a large body of customers. 

\begin{figure}
\begin{centering}
\includegraphics[width=3.3in]{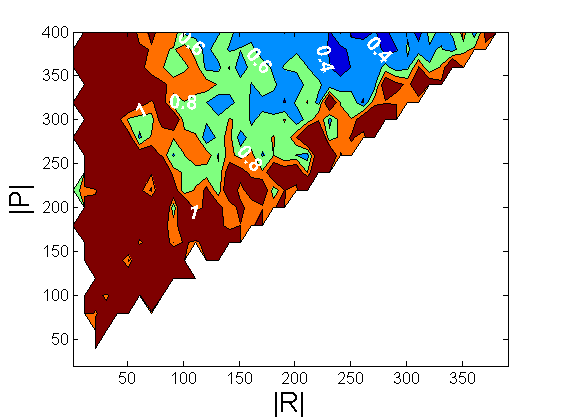}
\par\end{centering}
\caption{Ratio of new to old solution times after limiting the number of transfer
stations\label{fig:Stations}}
\end{figure}

\FloatBarrier

\subsection{Selecting a time interval}

In the numerical experiments in this paper, we discretized the study
time horizon into 1-minute time intervals. In this section, we study
the impact of using larger time intervals on the solution times, and
discuss when using larger time intervals is appropriate. 

The impact of using a larger time interval on solution times is not
obvious without running experiments. On one hand, using a larger time
interval will lead to smaller link sets produced by the pre-processing
procedure. On the other hand, using a larger time interval causes
a higher degree of temporal proximity among trips, which can
lead to higher solution times. Figure \ref{fig:Solution-times-dt5}
shows the solution times for the 420 randomly generated problems,
solved using the decomposition algorithm with 5-minute time intervals.
This figure suggests that utilizing larger time intervals leads to
lower solution times. It should be noted that as participants'
maximum ride times become larger, so do the size of the link sets, and therefore
the savings in solution times due to incorporating higher time intervals
become more significant. Using larger time intervals, in addition to yielding lower solution times, can make it less probable for travel time uncertainties to cause individuals to miss their rides, especially during transfers.

Another issue that should be discussed is under what circumstances it is appropriate to use larger time intervals. In practice, individuals
tend to report their schedules in five-minute increments at the finest,
i.e., it is not typical to hear an individual planning to leave no
later than 2:21 P.M., and therefore in general using larger time intervals
(e.g., 5 min intervals) seems appropriate. For a ridesharing system
that is integrated with a high-frequency transit system with minute-level
time schedule, however, using smaller time intervals might be more
suitable.

\begin{figure}
\begin{centering}
\includegraphics[width=3.3in]{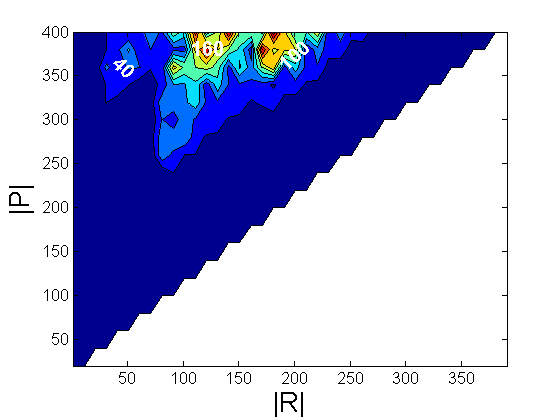}
\par\end{centering}

\caption{Solution times with 5-min time intervals\label{fig:Solution-times-dt5}}
\end{figure}

\subsection{Heuristic solutions\label{sub:Heuristic-solutions}}

In each iteration of the decomposition algorithm, a solution to the
original ride-matching problem is available. In section \ref{sub:Bounds},
we saw that we can obtain a heuristic feasible solution in each iteration
of the decomposition algorithm by solving a set packing problem. In
this section we study the impact of stopping the decomposition algorithm
at certain time limits and retrieving the best heuristic solution.
The results are displayed in Figure \ref{fig:Heuristic}. 

Figure \ref{fig:Heuristic} suggests that the additional number of
served riders obtained from allowing the computations to continue beyond
240 seconds is negligible. Comparing Figures \ref{fig:Heuristic}
and \ref{fig:Multi-Hop} suggests that even at a computational cut-off
time of 120 seconds, the number of served riders is higher than or
comparable to the number of served riders in the other four ride-matching
methods described in section \ref{sub:Value-of_Multihop}. In fact,
for several of the problems, even 60 seconds is sufficient. The fact
that computations can be accomplished in a mere 2-minute period again
validates the applicability of the optimizations schemes in dynamic
settings.

\begin{figure}
\begin{centering}
\includegraphics[width=3.5in]{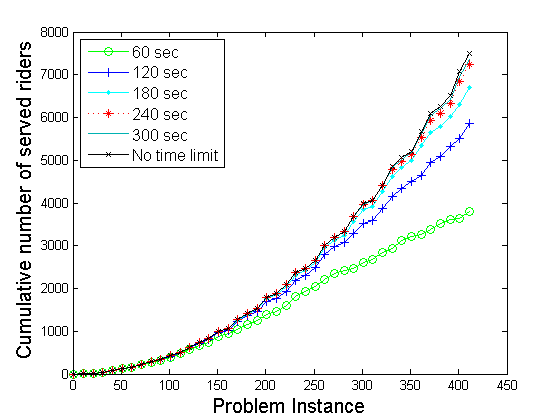}
\par\end{centering}

\caption{Quality of heuristic solutions\label{fig:Heuristic}}
\end{figure}

\section{Conclusion}

In this paper, we formulated the multi-hop P2P ride-matching problem
as a binary program. The formulation has a very tight linear relaxation,
and its structure limits the necessity for branching only to very
specific cases. We proposed a pre-processing procedure to reduce the
size of the input sets to this optimization problem, and an exact
and finite decomposition algorithm to solve the problem by iteratively
solving smaller sub-problems. Sub-problems in each iteration are independent
of each other, making it possible for computations to be done in parallel.
The computational savings achieved by the proposed algorithm can make
it applicable for real-time applications. Furthermore, the fact that
the proposed algorithm allows for as many transfers as necessary makes
it suitable for use in multi-modal networks.

We conducted extensive numerical experiments to study the impact of
multi-hop ridesharing on the number of served riders by comparing
matching rates of five different matching algorithms with different
degrees of flexibility. These experiments demonstrated the higher efficiency and throughput of a multi-hop ridesharing system. Another interesting
observation was that the value of a multi-hop solution becomes more
prominent when driver routes are determined by the system, and are
not pre-specified. Furthermore, we conducted sensitivity study over
different parameters, including the ratio of riders and drivers in
the system, the spatiotemporal distribution of trips, and the maximum
ride times requested by individuals.

We discussed different approaches, including clustering of trips in
time and space and extracting heuristic solutions in each iteration
of the algorithm, to make the proposed methodology applicable to dynamic,
real-life sized problems. It is hoped that the algorithms presented
and their variants will be applied in the near future to a variety
of socially and economically beneficial ridesharing systems, and perhaps
even to the scenarios of shared rides in automated or driverless vehicles
that are further distant in the future. 


\bibliographystyle{chicago}
\bibliography{refs}

\appendix

\section{Proof of Proposition\label{sec:Proof-of-Proposition}}
\begin{prop}
Number of connections for rider $r$ can be calculated using term
$\sum_{d\in D}u_{r}^{d}-1$. \end{prop}
\begin{proof}
If driver $d$ carries rider $r$ on any link, then $u_{r}^{d}=1$.
Therefore, the term $\sum_{d\in D}u_{r}^{d}-1$ provides the number of connections, only
if a driver does not pick up a rider multiple times. Without loss
of generality, we use the example in Figure \ref{fig:Undesireable_match}
to show by contradiction that such a scenario cannot happen. Figure
\ref{fig:Undesireable_match} shows a rider's itinerary.
On the first and third links, the rider is traveling with $d_{1}$,
and on the second link, he/she is traveling with $d_{2}$. If $d_{1}$
travels on both $\ell_{1}$ and $\ell_{3}$, at some point $d_{1}$
must have traveled from node 2 to node 3, and the rider could have
accompanied him on that ride too, reducing the number of transfers
from 2 to 0. The term $-W_{r}\sum u_{r}^{d}$ in the maximization
objective function ensures that the route with zero transfers is
selected as the optimal solution. Alternatively, one could leave out
this term from the objective function, and use post-processing to
refine the solution. 
\end{proof}
\begin{figure}[h]
\begin{centering}
\includegraphics[width=2.9in]{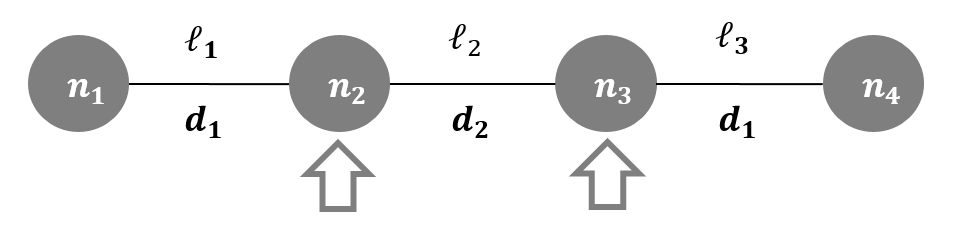}
\par\end{centering}

\caption{An undesirable match\label{fig:Undesireable_match}}
\end{figure}

\begin{prop}
(\ref{eq:ident_driver_1}) and (\ref{eq:ident_driver_2}) register
all drivers who collectively form each rider\textquoteright s itinerary. \end{prop}
\begin{proof}
From (\ref{eq:ident_driver_1}) we have $u_{r}^{d}\;\ge\;y_{\ell}^{rd}$.
If $y_{\ell}^{rd}=1$, $u_{r}^{d}$ is forced to be 1. If $y_{\ell}^{rd}=0$,
$u_{r}^{d}$ can take either 0 or 1. Constraint set (\ref{eq:ident_driver_2})
ensures that in the latter case, $u_{r}^{d}$ takes its lower bound. 
\end{proof}

\section{Decomposition Algorithm\label{sec:Decomposition-Algorithm}}

We start the algorithm by setting the iteration counter, $i$, to
1 (line 2), and the sub-problem counter, $k$, to 0 (line 3). Each rider forms a new sub-problem in the first iteration of the algorithm, and the type of all sub-problems is set as active (lines 4-9). In Algorithm 2, $R_{i}^{k}$ denotes the set of riders in sub-problem $k$ of iteration $i$, and $SP_{i}$ denotes the set of sub-problems in iteration $i$. $Active(k)$
indicates whether the type of sub-problem $k$ is active (1) or not (0). 

The second step of the algorithm involves solving the set of active sub-problems in the current iteration. We keep in set $D_{i}^{k}$ the list of all matched drivers in sub-problem $k$ of iteration $i$ (line 13), and in set $D_{i}$ the list of all drivers who have been matched in iteration $i$ (line 14). Furthermore, we use sets $Y_{i,k},$ $X_{i,k}^{d}$, and $Z_{i,k}^{r}$
to keep track of the itineraries of all riders, driver $d$ ($\forall d\in D)$,
and whether rider $r$ $(\forall r\in R)$ has been matched in sub-problem
$k$ of iteration $i$, respectively (lines 15-17).

In step 3 we check the stopping criterion of the algorithm. We start by setting the value of the indicator variable $Terminate(d)$ to zero (line 21). If driver $d$ is shown to have the same itinerary under all sub-problems, the value of the indicator variable $Terminate(d)$ will be set to 1 (lines 22-24). If the value of $Terminate(d)$ becomes 1 for all drivers, then the union of solutions to the sub-problems will yield the optimal solution to the original optimization problem $(X^{*})$ (lines 26-29).\\
\\
\begin{tabular}{l}
\hline 
\textbf{Algorithm 2} The decomposition algorithm\tabularnewline
\hline 
\textbf{}%
\begin{minipage}[t]{1\columnwidth}%
$01\quad$ \textbf{Step 1. Initialize}

$02\quad$ $i\leftarrow 1$

$03\quad$ $k\leftarrow 0$

$04\quad$ For each rider $r \in R$

$05\quad$ $\qquad k\leftarrow k+1.$

$06\quad$ $\qquad$Set $R_{i}^{k}=\{r\}.$

$07\quad$ $\qquad$Set $Active(k)\leftarrow1.$

$08\quad$ End For

$09\quad$ Set $SP_{i}=\{1,...,|R|\}.$

$10\quad$ \textbf{Step 2. Solve the (active) sub-problems}

$11\quad$ For each sub-problem $k\in SP_{i}$ for which $Active(k)=1$

$12\quad$ $\qquad$Solve sub-problem $k$.

$13\quad$ $\qquad$Let $D_{i}^{k}=\{d\in D\;|\;\sum_{r\in R_{i}^{k},\ell\in L_{rd}}y{}_{\ell}^{rd*}\geq1\}.$

$14\quad$ $\qquad$Let $D_{i}=\bigcup_{k\in SP_{i}}D_{i}^{k}.$

$15\quad$ $\qquad$Let $Y_{i,k}=\big(y{}_{\ell}^{rd*}\big)_{r\in R_{i}^{k},d\in D_{i}^{k},\ell\in L_{rd}}.$ 

$16\quad$ $\qquad$Let $X_{i,k}^{d}=\big(x_{\ell}^{d*}\big)_{\ell\in L_{d}},\;\forall d\in D_{i}^{k}.$

$17\quad$ $\qquad$Let $Z_{i,k}^{r}=z_{r}^{*},\;\forall r\in R_{i}^{k}.$

$18\quad$ End For

$19\quad$ \textbf{Step 3. Termination Criterion}

$20\quad$ For each driver $d\in D_{i}$

$21\quad \qquad Terminate(d) \leftarrow0$.

$22\quad$ $\qquad$If $X_{i,m}^{d}=X_{i,n}^{d}$ for all $(m,n)\in D_{i}^{m}\times D_{i}^{n}$

$23\quad$ $\qquad\qquad Terminate(d)\leftarrow1.$

$24\quad$ $\qquad$End If

$25\quad$ End For

$26\quad$ If $\sum_{d\in D_{i}}Terminate(d)=|D|$ then 

$27\quad \qquad$ Stop. 

$28\quad$ Final optimal solution $X^{*}=U_{k\in SP_{i}}\big(Y_{i,j},\cup_{d\in D}X_{i,k}^{d},\cup_{r\in R}Z_{i,k}^{r}\big)$.

$29\quad$ Else 

$30\quad$ $\qquad$Continue to Step 4.

$31\quad$ End If

\end{minipage}\tabularnewline
\hline 
\end{tabular}\\
\\
If the algorithm does not terminate in step 3, we move forward to
step 4 to form the set of sub-problems for the next iteration (line 33). We
start step 4 by defining set $R^{\prime}$ to include all riders that
need to be allocated to sub-problems in the new iteration, and initialize
it with $R$ (line 35). Next, we go through the list of matched drivers in the
previous iteration (in no particular order). For each driver, we find
the riders in set $R^{\prime}$ that have the driver on their itinerary.
These riders (if belonging to different sub-problems) will from a
new sub-problem in the current iteration (line 42), and are removed from set
$R^{\prime}$ (line 43). We mark all such sub-problems $k$ as applicable by
setting the value of indicator $Applicable(k)$ to 1 (line 43).

At this point, there could still be riders that do not belong to any
sub-problems in the new iteration. Next, we go through the sub-problems
from the previous iteration one by one (lines 47-58). For a given sub-problem $j$
from the previous iteration, we find all riders $r$ who do not belong
to any sub-problems in the current iteration (i.e., $r\in R^{\prime}$).
If such riders exist, they will from a new sub-problem in the current
iteration. If it turns out that we are importing an entire sub-problem
from the previous iteration to the current iteration, then this sub-problem
does not need to be solved again, and so we set the applicability
indicator of the sub-problem to zero (lines 49-51). Otherwise, if we are importing
only a part of a sub-problem from the previous iteration as a new
sub-problem to the current iteration, then the solution from the previous
iteration is not necessarily optimal anymore, and so we set the applicability
indicator of the new sub-problem to 1 (lines 53-55). We next move to Step 5.\\
\\
\begin{tabular}{l}
\hline 
\textbf{Algorithm 2 (continued)} The decomposition algorithm\tabularnewline
\hline 
\textbf{}%
\begin{minipage}[t]{1\columnwidth}%
$32\quad$ \textbf{Step 4. Form the set of sub-problems for the next iteration}

$33\quad$ $i\leftarrow i+1$

$34\quad$ $k\leftarrow0$

$35\quad$ Let $R^{\prime}=R$.

$36\quad$ For each driver $d\in D_{i-1}$ such that there are at least two sub-problems
$(m,n)\in SP_{i-1} \times SP_{i-1}$ 

$37\quad$  with $d\in D_{i-1}^{m}$ and $d\in D_{i-1}^{n}$

$38\quad$ $\qquad$Let $R_{temp}=$ set of all riders $r\in R^{\prime}$ that
have driver $d$ on their optimal path 

$39\quad$ $\qquad$in any sub-problems $k\in SP_{i-1}$.

$40\quad$ $\qquad$If $R_{temp}\not=\emptyset$ then 

$41\quad$ $\qquad\qquad$Form a new sub-problem: $k\leftarrow k+1$. 

$42\quad$ $\qquad\qquad$Set $R_{i}^{k}=R_{temp}$. 

$43\quad$ $\qquad\qquad$Update $R^{\prime}\leftarrow R^{\prime}\backslash R_{temp}$. 

$44\quad$ $\qquad\qquad$Set $Applicable(k)\leftarrow1$. 

$45\quad$ $\qquad$End If

$46\quad$ End For

$47\quad$ For $j\in SP_{i-1}$

$48\quad$ $\qquad$If $|R_{i-1}^{j}\cap R^{\prime}|=|R_{i-1}^{j}|$ then 

$49\quad$ $\qquad\qquad$Form a new sub-problem: $k\leftarrow k+1$. 

$50\quad$ $\qquad\qquad$Set $R_{i}^{k}=R_{i-1}^{j}$. 

$51\quad$ $\qquad\qquad$Set $Applicable(k)\leftarrow0$.

$52\quad$ $\qquad$Elseif $R_{i-1}^{j}\cap R^{\prime}\neq\emptyset$ and $|R_{i-1}^{j}\cap R^{\prime}|<|R_{i-1}^{j}|$ 

$53\quad$ $\qquad\qquad$Form a new sub-problem: $k\leftarrow k+1$. 

$54\quad$ $\qquad\qquad$Set $R_{i}^{k}=R_{i-1}^{j}\cap R^{\prime}$. 

$55\quad$ $\qquad\qquad$Set $Applicable(k)\leftarrow1$.

$56\quad$ $\qquad$End If

$57\quad$ End For

$ $%
\end{minipage}\tabularnewline
\hline 
\end{tabular}\\
\\
In step 5, we prevent the algorithm from looping between iterations.
After forming the set of sub-problems for the current iteration, we
compare these sub-problems with sub-problems from previous iterations.
If it turns out that two iterations have the exact same set of sub-problems,
we form a new intermediate sub-problem in the current iteration by
finding two sub-problems from the previous iteration such that each
of these two sub-problems includes a subset of riders in a sub-problem
in the current iteration (lines 62-63). These two sub-problems form a new sub-problem
in the current iteration (line 64). Naturally, riders in this newly formed sub-problem
are eliminated from their original sub-problems (line 65), and the newly formed
intermediate sub-problem is marked as applicable (line 66).

The last step of the algorithm is to find the active sub-problems:
those whose solutions cannot be derived from the solutions of previously
solved sup-problems. There are two cases where sub-problems are not
active: first, if the exact same sub-problem has been solved before,
in which case the solution is readily available (lines 78-82); second, if a sub-problem
$k$ is a union of a set of sub-problems from a previous iteration,
and the solutions to these sub-problems do not have any conflicts
in terms of itineraries of drivers. In this case, the solution to
sub-problem $k$ would be the union of the solutions to this set of
sub-problems (lines 87-90). Next, we go back to Step 2, where we solve the active sub-problems (line 96).\\
\\
\begin{tabular}{l}
\hline 
\textbf{Algorithm 2 (continued)} The decomposition algorithm\tabularnewline
\hline 
\begin{minipage}[t]{1\columnwidth}%
$58\quad$ \textbf{Step 5. Check for repeating patterns and form intermediate
sub-problems}

$59\quad$ For $j\in1$ to $i-2$

$60\quad$ $\qquad$If iterations $i$ and $j$ have the same set of sub-problems,
then

$61\quad$ $\qquad\qquad$For sub-problem $k\in SP_{i}$ 

$62\quad$ $\qquad$$\qquad\qquad$Find two sub-problems $m,n\in SP_{i-1}$ such
that they each share 

$63\quad$ $\qquad\qquad\qquad$at least one rider with $R_{i}^{k}$.

$64\quad$ $\qquad$$\qquad\qquad$Add a sub-problem $k$ to $SP_{i}$ with combined
set of riders $R_{i}^{m}\cup R_{i}^{n}$.

$65\quad$ $\qquad\qquad\qquad$Delete all riders in $R_{i}^{m}\cup R_{i}^{n}$
from sub-problems in $SP_{i}$.

$66\quad$ $\qquad\qquad\qquad$Set $Applicable(k)\leftarrow1.$

$67\quad$ $\qquad$$\qquad\qquad$Exit For loop.

$68\quad$ $\qquad\qquad$End For

$69\quad$ $\qquad$End If

$70\quad$ End For

$71\quad$ \textbf{Step 6. Identify active sub-problems}

$72\quad$ For each sub-problem $k\in SP_{i}$ such that $Applicable(k)=1$

$73\quad$ $\qquad$Set $Active(k)\leftarrow1$

$74\quad$ $\qquad$For each iteration $j=1$ to $i-1$

$75\quad$ $\qquad\qquad$For each sub-problem $n\in SP_{j}$

$76\quad$ $\qquad\qquad\qquad$While $Active(k)=1$

$77\quad$ $\qquad\qquad\qquad\qquad$If $R_{i}^{k}=R_{j}^{n}$ then

$78\quad$ $\qquad\qquad\qquad\qquad\qquad$The solution to $k$ is already available: 

$79\quad$ $\qquad\qquad\qquad\qquad\qquad$Set $Y_{i,k}=Y_{j,n}$. 

$80\quad$ $\qquad\qquad\qquad\qquad\qquad$Set $X_{i,k}^{d}=X_{j,n}^{d}$$\;\forall d\in D_{j}^{n}.$ 

$81\quad$ $\qquad\qquad\qquad\qquad\qquad$Set $Z_{i,k}^{r}=Z_{j,n}^{r}\;\forall r\in R_{n}^{j}$. 

$82\quad$ $\qquad\qquad\qquad\qquad\qquad$Set $Active(k)\leftarrow0$.

$83\quad$ $\qquad\qquad\qquad\qquad$Elseif there exist a set of sub-problems
$N\in SP_{j}$ such that

$84\quad$$\qquad\qquad\qquad\qquad$$R_{i}^{k}=\cup_{n\in N}R_{j}^{n}$, and
the solutions to sub-problems in set $N$ 

$85\quad$ $\qquad\qquad\qquad\qquad$have no conflict (i.e., $\forall m,n\in N$
and $\forall d\in D_{j}^{m}\cap D_{j}^{n}\;:\;X_{j,m}^{d}=X_{j,n}^{d}$
) 

$86\quad$ $\qquad\qquad\qquad\qquad\qquad$The solution to $k$ is the union
of the solutions to $n\in N$:

$87\quad$ $\qquad\qquad\qquad\qquad\qquad$Set $Y_{i,k}=\big(Y_{j,n}\big)_{n\in N}$. 

$88\quad$ $\qquad\qquad\qquad\qquad\qquad$Set $X_{i,k}^{d}=X_{j,n}^{d}$$\;\forall d\in\bigcup_{n\in N}D_{j}^{n}.$ 

$89\quad$ $\qquad\qquad\qquad\qquad\qquad$Set $Z_{i,k}^{r}=Z_{j,n}^{r}\;\forall r\in\bigcup_{n\in N}R_{n}^{j}$. 

$90\quad$ $\qquad\qquad\qquad\qquad\qquad$Set $Active(k)\leftarrow0$.

$91\quad$ $\qquad\qquad\qquad\qquad$End If

$92\quad$ $\qquad\qquad\qquad$End While

$93\quad$ $\qquad\qquad$End For

$94\quad$ $\qquad$End For

$95\quad$ End For

$96\quad$ Go to Step 2

$ $%
\end{minipage}\tabularnewline
\hline 
\end{tabular}
\\
\newpage

Below we have presented a summary of Algorithm 2, and have elaborated on how one should move between steps of this algorithm. In addition to the summary below, the algorithm flowchart in Figure \ref{fig:The-decomposition-algorithm} showcases the connections between different steps of the algorithm.
\begin{mdframed}

\textbf{Step 1}: Set the number of iterations to 1. Initialize the decomposition algorithm by forming the first iteration's set of sub-problems (i.e., one sub-problem for each rider.) Set the state of all sub-problems as active. Go to \textbf{Step 2}.\\

\noindent\textbf{Step 2}: Solve the active sub-problems. Go to \textbf{Step 3}.\\

\noindent\textbf{Step 3}: Check the termination Criterion. If the criterion is satisfied, then \textbf{stop}. Otherwise, go to \textbf{Step 4}.\\

\noindent\textbf{Step 4}: Increment the number of iterations by 1. For this new iteration construct the set of sub-problems. Go to \textbf{Step 5}.\\

\noindent\textbf{Step 5}: Look for repeating patterns. If such patterns exist, adjust the iteration's sub-problems by creating intermediate sub-problems. Go to \textbf{Step 6}. \\

\noindent\textbf{Step 6}: Identify active sub-problems in the current iteration. Go to \textbf{Step 2}.

\end{mdframed}


\section{Revised version of the P2P multi-hop matching problem\label{sec:Revised-version-optimization}}

Each sub-problem $k$ in iteration $i$ of the decomposition algorithm
needs to be solved using the optimization problem in model (\ref{equ:determinstic-1}).
In the interest of simplicity of notation, let us denote $R_{i}^{k}$
and $D{}_{i}^{k}$, i.e., the rider and drivers sets in sub-problem
$k$ in iteration $i$, by $R_{\kappa}$ and $D_{\kappa}$, respectively.

\begin{subequations}\label{equ:determinstic-1}

\begin{align}
\mbox{Max}\quad & \;\;\sum_{r\in R_{\kappa}}z_{r}\:\:\:-\sum_{r\in R_{\kappa}}W_{r}\sum_{d\in D_{\kappa}:(r,d)\in M}u_{r}^{d}\label{eq:Obj-1}\\
\mbox{Subject to:}\quad & \sum_{^{_{s_{i}=OS_{d}}^{\:\:\:\ell\in L_{d}:}}}x_{\ell}^{d}\:\:\:\:-\sum_{^{_{s_{j}=OS_{d}}^{\:\;\:\ell\in L_{d}:}}}x_{\ell}^{d}=1 &  & \forall d\in D_{\kappa}\label{eq:d_orig-1}\\
 & \sum_{^{_{s_{j}=DS_{d}}^{\:\:\:\ell\in L_{d}:}}}x_{\ell}^{d}\:\:\:\:-\sum_{^{_{s_{i}=DS_{d}}^{\:\:\:\ell\in L_{d}:}}}x_{\ell}^{d}=1 &  & \forall d\in D_{\kappa}\label{eq:d_dest-1}\\
 & \!\!\!\!\!\!\!\sum_{_{\ell=(t_{i},s_{i},t,s)\in L_{d}}^{\qquad t_{i},s_{i}}}\!\!\!\!\!x_{\ell}^{d}\:\:\:=\!\!\!\!\!\!\sum_{_{\ell=(t,s,t_{j},s_{j})\in L_{d}}^{\qquad t_{j},s_{j}}}\!\!\!\!\!x_{\ell}^{d} &  & \!\!\!\begin{array}{l}
\forall d\in D_{\kappa}\\
\forall t\in T_{d}\\
\forall s\in G_{d}\backslash\{OS_{d}\cup DS_{d}\}
\end{array}\label{eq:d_balance-1}\\
 & \!\!\!\sum_{^{_{\quad(r,d)\in M}^{\quad d\in D_{\kappa}^{\prime}:}}}\!\!\!\sum_{^{_{\quad s_{i}=OS_{r}}^{\quad\ell\in L_{rd}:}}}\!\!\!\!\!\!y_{\ell}^{rd}-\sum_{^{_{\quad(r,d)\in M}^{\quad d\in D_{\kappa}^{\prime}:}}}\!\!\!\sum_{^{_{\quad s_{j}=OS_{r}}^{\quad\ell\in L_{rd}:}}}\!\!\!\!\!\!y_{\ell}^{rd}=z_{r} &  & \forall r\in R_{\kappa}\label{eq:r_orig-1}\\
 & \!\!\!\sum_{^{_{\quad(r,d)\in M}^{\quad d\in D_{\kappa}^{\prime}:}}}\!\!\!\sum_{^{_{\quad s_{j}=DS_{r}}^{\quad\ell\in L_{rd}:}}}\!\!\!\!\!\!y_{\ell}^{rd}-\sum_{^{_{\quad(r,d)\in M}^{\quad d\in D_{\kappa}^{\prime}:}}}\!\!\!\sum_{^{_{\quad s_{i}=DS_{r}}^{\quad\ell\in L_{rd}:}}}\!\!\!\!\!\!y_{\ell}^{rd}=z_{r} &  & \forall r\in R_{\kappa}\label{eq:r_dest-1}\\
 & \!\!\!\sum_{^{_{\quad(r,d)\in M}^{\quad d\in D_{\kappa}^{\prime}:}}}\!\!\!\sum_{^{\:\:\:_{\ell=(t_{i},s_{i},t,s)\in L}^{\:\:\:\:\:\:\:\:\:t_{i},s_{i}:}}}\!\!\!\!\!y_{\ell}^{rd}=\!\!\!\sum_{^{_{\quad(r,d)\in M}^{\quad d\in D_{\kappa}^{\prime}:}}}\!\!\!\sum_{^{_{\:\:\:\ell=(t,s,t_{j},s_{j})\in L}^{\:\:\:\:\:\:\:\:\:\:\:\:\:t_{j},s_{j}:}}}\!\!\!\!\!y_{\ell}^{rd} &  & \!\!\!\begin{array}{l}
\forall r\in R_{\kappa}\\
\forall t\in T_{r}\\
\forall s\in G_{r}\backslash\{OS_{r}\cup DS_{r}\}
\end{array}\label{eq:r_balance-1}\\
 & \!\!\!\!\!\!\!\!\!\!\!\!\sum_{^{_{\quad\quad(r,d)\in M,\ell\in L_{rd}}^{\quad\quad\:\:\:\:r\in R_{\kappa}:}}}\!\!\!\!\!\!\!\!\!y_{\ell}^{rd}\leq C_{d}x_{\ell}^{d} &  & \!\!\!\begin{array}{l}
\forall d\in D_{\kappa}\\
\forall\ell\in L_{d}
\end{array}\label{eq:capacity-1}\\
 & \;\;u_{r}^{d}\geq y_{\ell}^{rd} &  & \!\!\!\begin{array}{l}
\begin{array}{l}
\forall(r,d)\in M\\
\forall\ell\in L_{rd}
\end{array}\end{array}\label{eq:ident_driver_1-1}\\
 & \;\;u_{r}^{d}\leq\sum_{\ell\in L_{rd}}y_{\ell}^{rd} &  & \begin{array}{l}
\forall(r,d)\in M\end{array}\label{eq:ident_driver_2-1}\\
 & \!\!\!\sum_{^{_{\quad(r,d)\in M}^{\quad d\in D_{\kappa}^{\prime}:}}}y_{\ell}^{rd}-1\leq V_{r} &  & \:\:\forall r\in R_{\kappa}\label{eq:transfer-1}
\end{align}

\end{subequations}

\end{document}